\documentclass[aip,jmp,amsmath,amssymb,author-year,reprint]{revtex4-1}
\usepackage{amsthm}
\usepackage{graphicx}% Include figure files
\usepackage{bm}% bold math

\newcommand{\abs}[1]{\left \lvert #1 \right\rvert} % for absolute value
\newcommand{\inner}[2]{\left< #1, #2 \right>} % for average

\renewcommand{\vec}[1]{\boldsymbol{#1}  } % for vectors

\renewcommand{\d}[2]{\frac{d #1}{d #2}} % for derivatives
\renewcommand{\d}[2]{\frac{d #1}{d #2}} % for derivatives
\newcommand{\grad}[1]{\vec{\nabla} #1} % for gradient
\newcommand{\lap}{\nabla^2} % for laplacian 
\newcommand{\laptwo}{\nabla^4} % for laplacian squared 
\newcommand{\unit}[1]{\hat{\vec{#1}}}

\newcommand{\ov}[1]{{#1}^\ast }

\newcommand{\C}{\mathbb{C}}
\newcommand{\intz}{\int_{z_1}^{z_2}}

\newcommand{\Ac}{\mathcal{A}}
\newcommand{\Bc}{\mathcal{B}}
\newcommand{\Cc}{\mathcal{C}}

\newcommand{\Lc}{\mathcal{L}}

\newtheorem{theorem}{Theorem}[section]
\newtheorem{proposition}[theorem]{Proposition}
\newtheorem{lemma}[theorem]{Lemma}

\newtheorem{definition}[theorem]{Definition}

\begin{document}

%\preprint{}

\title{Normal modes with boundary dynamics in geophysical fluids} %Title of paper



\author{Houssam Yassin}
\email[]{hyassin@princeton.edu}
%\thanks{}
\affiliation{Program in Atmospheric and Oceanic Sciences, Princeton University,
Princeton, NJ 08544, USA}


\date{\today}

\begin{abstract}
Three-dimensional geophysical fluids support both internal and boundary-trapped waves. To obtain the normal modes in such fluids we must solve a differential eigenvalue problem for the vertical structure (for simplicity, we only consider horizontally periodic domains). If the boundaries are dynamically inert (e.g., rigid boundaries in the Boussinesq internal wave problem, flat boundaries in the quasigeostrophic Rossby wave problem) the resulting eigenvalue problem typically has a Sturm-Liouville form and the properties of such problems are well-known. However, when restoring forces are also present at the boundaries, then the equations of motion contain a time-derivative in the boundary conditions and this leads to an eigenvalue problem where the eigenvalue correspondingly appears in the boundary conditions. In certain cases, the eigenvalue problem can be formulated as an eigenvalue problem in the Hilbert space $L^2\oplus \C$ and this theory is well-developed. Less explored is the case when the eigenvalue problem takes place in a Pontryagin space, as in the Rossby wave problem over sloping topography. This article develops the theory of such problems and explores the properties of wave problems with dynamically-active boundaries. The theory allows us to solve the initial value problem for quasigeostrophic Rossby waves in a region with sloping bottom (we also apply the theory to two Boussinesq problems with a free-surface). For a step-function perturbation at a dynamically-active boundary, we find that the resulting time-evolution consists of waves present in proportion to their projection onto the dynamically-active boundary. 
\end{abstract}


\maketitle %\maketitle must follow title, authors, abstract and \pacs

\section{Introduction}\label{S-intro}

An important tool in the study of wave motion near a stable equilibrium is the separation of variables. When applicable, this elementary technique transforms a linear partial differential equation into an ordinary differential eigenvalue problem for each coordinate \cite[e.g.,][]{hillen_partial_2012}. Upon solving the differential eigenvalue problems, one obtains the normal modes of the physical system. The normal modes are the fundamental wave motions for the given restoring forces, each mode represents an independent degree of freedom in which the physical system can oscillate, and any solution of the wave problem may be written as a linear combination of these normal modes. 

To derive the normal modes, we must first linearize the dynamical equations of motion about some equilibrium state. We then encounter linearized restoring forces of two kinds: 
\begin{itemize}
	\item[\emph{1.}] volume-permeating forces experienced by fluid particles in the interior, and
	\item[\emph{2.}] boundary-confined forces only experienced by fluid particles at the boundary.
\end{itemize} 
Examples of volume-permeating forces include the restoring forces resulting from continuous density stratification and continuous volume potential vorticity gradients. These restoring forces respectively result in internal gravity waves \citep{sutherland_internal_2010} and Rossby waves \citep{Vallis2017}. Examples of boundary-confined restoring forces include the gravitational force at a free-surface (i.e., at a jump discontinuity in the background density), forces arising from gradients in surface potential vorticity \citep{schneider_boundary_2003}, and the molecular forces giving rise to surface tension. These restoring forces respectively result in surface gravity waves \citep{sutherland_internal_2010}, topographic/thermal waves \citep{hoskins_use_1985}, and capillary waves \citep{lamb_hydrodynamics_1975}. 

In the absence of  boundary-confined restoring forces, we can often apply Sturm-Liouville theory \cite[e.g.,][]{hillen_partial_2012,zettl_sturm-liouville_2010} to the resulting eigenvalue problem. We thus obtain a countable infinity of waves whose vertical structures form a basis of $L^2$, the space of square-integrable functions (see \S \ref{math-section}), and, given some initial vertical structure, we know how to solve for the subsequent time-evolution as a linear combination for linearly independent waves. Moreover, a classic result of Sturm-Liouville theory is that the $n$th mode has $n$ internal zeros.

In the presence of boundary-confined restoring forces, the governing equations have a time-derivative in the boundary conditions. The resulting eigenvalue problem correspondingly contains the eigenvalue parameter in the boundary conditions. Sturm-Liouville theory is inapplicable to such problems.

In this article, we present a general method for solving these problems by delineating a generalization of Sturm-Liouville theory. Some consequences of this theory are the following. There is a countable infinity of waves whose vertical structures form a basis of $L^2\oplus\C^s$, where $s$ is the number of dynamically-active boundaries; thus, each boundary-trapped wave, in mathematically rigorous sense, provides an additional degree of freedom to the problem. The modes satisfy an orthogonality relation involving boundary terms, the modes may have a negative norm, and the modes may have finite jump discontinuities at dynamically-active boundaries (although the \emph{solutions} are always continuous, see \S \ref{S-Boussinesq-rot}). When negative norms are possible (as in quasigeostrophic theory), there is a new expression for the Fourier coefficients that one must use to solve initial value problems [see equation \eqref{expansion}]. We can also expand boundary step-functions (representing some boundary localized perturbation) as a sum of modes. Moreover, the $n$th mode may not have $n$ internal zeros; indeed, depending on physical parameters in the problem, two or three linearly independent modes with an identical number of internal zeros may be present.  

%%% NEW ADD
We also show that the eigenfunction expansion of a function is term-by-term differentiable, with the derivative series converging uniformly on the whole interval, regardless of the boundary condition the function satisfies at the dynamically-active boundaries. This property is in contrast with a traditional Sturm-Liouville eigenfunction expansion where the term-by-term derivative converges uniformly only if the function satisfies the same boundary condition as the eigenfunctions.

We apply the theory to three geophysical wave problems. The first is that of a Boussinesq fluid with a free-surface; we find that the $n$th mode has $n$ internal zeros. The second example is that of a rotating Boussinesq fluid with a free-surface where we assume that the stratification suppresses rotational effects in the interior but not at the upper boundary. We find that there are two linearly independent modes with $M$ internal zeros, where the integer $M$ depends on the ratio of the Coriolis parameter to the horizontal wavenumber, and that the eigenfunctions have a finite jump discontinuity at the upper boundary. The third application is to a quasigeostrophic fluid with a sloping lower boundary. We find that modes with an eastward phase speed have a negative norm whereas modes with a westward phase speed have a positive norm (the sign of the norm has implications for the relative phase of a wave and for series expansions). Moreover, depending on the propagation direction, there can be two linearly independent modes with no internal zeros. For all three examples, we outline the properties of the resulting series expansions and provide the general solution. We also consider the time-evolution resulting from a vertically localized perturbation at a dynamically-active boundary; we idealize such a perturbation as a boundary step-function. The step-function perturbation induces a time-evolution in which the amplitude of each constituent wave is proportional to the projection of that wave onto the boundary. \label{qg-time-evolution}

To our knowledge, most of the above results cannot be found in the literature [however, the gravity wave orthogonality relation has been noted before, e.g., \cite{gill_atmosphere-ocean_2003} and \cite{kelly_vertical_2016} for the hydrostatic case and  \cite{olbers_internal_1986} and \cite{early_fast_2020} for the non-hydrostatic case]. For instance, we provide the only solution to the initial value problem for Rossby waves over topography in the literature [equation \eqref{qg-time-evolution}]. Moreover, many of the properties we discuss arise in practical problems in physical oceanography. The number of internal zeros of Rossby waves is also a useful quantity in observational physical oceanography [e.g., \cite{clement_vertical_2014} and \cite{de_la_lama_vertical_2016}]. In addition, the question of whether the quasigeostrophic baroclinic modes are complete is a controversial one. 
%%% NEW ADD
\cite{lapeyre_what_2009} has suggested that the baroclinic modes are incomplete because they assume a vanishing surface buoyancy anomaly. Consequently, \cite{smith_surface-aware_2012} address this issue by deriving an $L^2\oplus \C^2$ basis for quasigeostrophic theory. Yet many authors, citing completeness theorems from Sturm-Liouville theory, insist that the baroclinic modes are indeed complete and can represent all quasigeostrophic states \citep{ferrari_distribution_2010,lacasce_surface_2012,rocha_galerkin_2015}. This article shows that, by including boundary-confined restoring forces, we obtain a set of modes with additional degrees-of-freedom. These degrees-of-freedom manifest in the behaviour of eigenfunction expansions at the boundaries.
 %%% NEW REMOVE
%The issue of the differentiability of the resulting series expansions has been raised in various physical applications \citep{rocha_galerkin_2015,kelly_vertical_2016}. 
In addition, the distinction between $L^2$ and $L^2\oplus \C^s$ bases that we present here is useful for equilibrium statistical mechanical calculations where one must decompose fluid motion onto a complete set of modes \citep{bouchet_statistical_2012, venaille_catalytic_2012}.

The plan of the article is the following. We formulate the mathematical theory in \S \ref{math-section}. We then apply the theory to the two Boussinesq wave problems, in \S \ref{S-Boussinesq}, and to the quasigeostrophic wave problem, in \S \ref{S-QG}. We consider the time-evolution of a localized perturbation at a dynamically-active boundary in \S \ref{S-step-forcing}. We then conclude in \S \ref{S-conclusion}.

\section{The eigenvalue problem} \label{math-section}
	
In this section, we outline the theory of the differential eigenvalue problem
\begin{align} \label{EigenDiff}
	-(p \, \phi')' + q \, \phi &= \lambda \, r \, \phi \quad \text{for } \quad z\in \left(z_1,z_2\right)\\
	\label{EigenB1}
	- \left[a_1 \, \phi(z_1) - b_1 \, (p \, \phi')(z_1)\right] &= \lambda \left[c_1 \, \phi(z_1) - d_1 \, (p \, \phi')(z_1)\right] \\
	\label{EigenB2}
	- \left[a_2 \, \phi(z_2) - b_2 \, (p \, \phi')(z_2)\right] &= \lambda \left[c_2 \, \phi(z_2) - d_2 \, (p \, \phi')(z_2)\right] 	   	
\end{align}
where $p^{-1},q$, and $r$ are real-valued integrable functions; $a_i,b_i,c_i$, and $d_i$ are real numbers with $i\in\{1,2\}$; and where $\lambda \in \C$ is the eigenvalue parameter. We further assume that $p>0$ and $r>0$, that $p$ and $r$ are twice continuously differentiable, that $q$ is continuous,  and that $(a_i,b_i)\neq (0,0)$ for $i\in\{1,2\}$.
The system of equations \eqref{EigenDiff}\textendash\eqref{EigenB2} is an eigenvalue problem for the eigenvalue $\lambda \in \C$ and differs from a regular Sturm-Liouville problem in that $\lambda$ appears in the boundary conditions \eqref{EigenB1} and \eqref{EigenB2}.  That is, setting $c_{i}= d_{i} = 0$ recovers the traditional Sturm-Liouville problem. The presence of $\lambda$ as part of the boundary condition leads to some fundamentally new mathematical features that are the subject of this section and fundamental to the physics of this study. 

It is useful to define the two boundary parameters
    \begin{equation}\label{Di}
    	D_i = (-1)^{i+1} \left(a_i \, d_i - b_i \, c_i\right) \quad i=1,2. 
    \end{equation}
Just as the function $r$ acts as a weight for the interval $(z_1,z_2)$ in traditional Sturm-Liouville problems, the constants $D_i^{-1}$ will play analogous roles for the boundaries $z=z_i$ when $D_i\neq 0$.

\subsubsection*{Outline of the mathematics}

The right-definite case, when the $D_i \geq 0$ for $i\in\{1,2\}$, is well-known in the mathematics literature; most of the right-definite results in this section are due to \cite{evans_non-self-adjoint_1970}, \cite{walter_regular_1973}, and \cite{fulton_two-point_1977}. In contrast, the left-definite case, defined below, is much less studied. In this section, we generalize the right-definite results of \cite{fulton_two-point_1977} to the left-definite problem as well as provide an intuitive formulation \cite[in terms of functions rather than vectors, for a vector formulation see][]{fulton_two-point_1977} of the eigenvalue problem. 

In section \S \ref{S-Op-Form} we state the conditions under which we obtain real eigenvalues and a basis of eigenfunctions. We proceed, in  \S \ref{S-Prop-of-Eig}, to explore the properties of eigenfunctions and eigenfunction expansions. Finally, in \S \ref{S-oscillation}, we discuss oscillation properties of the eigenfunctions. Additional properties of the eigenvalue problem are found in appendix \ref{S-additional-properties} and a literature review, along with various technical proofs, is found in appendix \ref{S-math-app}.

\subsection{Formulation of the problem} \label{S-Op-Form}
\subsubsection{The functions space of the problem}
We denote by $L^2$ the Hilbert space of square-integrable ``functions" $\phi$ on the interval $(z_1,z_2)$ satisfying 
\begin{equation}
	\intz \abs{\phi}^2 \, r \, \mathrm{d}z < \infty.
\end{equation}
To be more precise, the elements of $L^2$ are not functions but rather equivalence classes of functions \cite[e.g.,][ section I.3]{reed_methods_1980}. Two functions, $\phi$ and $\psi$, are equivalent in $L^2$ (i.e., $\phi=\psi$ in $L^2$) if they agree in a mean-square sense on $[z_1,z_2]$,
\begin{equation}\label{L2-equal}
	\intz \abs{\phi(z) - \psi(z)}^2 \, r \, \mathrm{d}z = 0.
\end{equation}
Significantly, we can have $\phi=\psi$ in $L^2$ but $\phi\neq \psi$ pointwise. 
	
Furthermore, as a Hilbert space, $L^2$ is endowed with a positive-definite inner product
\begin{equation}\label{L2-inner}
	\inner{\phi}{\psi}_\sigma = \intz \ov{\phi}\, \psi \, \mathrm{d}\sigma = \intz \ov{\phi} \, \psi \ r\, \mathrm{d}z,
\end{equation}
where the symbol $\ov{\ }$ denotes complex conjugation and the measure $\sigma$ associated $L^2$ induces a differential element $\mathrm{d}\sigma = r\,\mathrm{d}z$ (see appendix \ref{S-additional-properties}). The positive-definiteness is ensured by our assumption that $r>0$ (i.e., $\inner{\phi}{\phi}_{\sigma} > 0$ for $\phi \neq 0$ when $r>0$).

It is well-known that traditional Sturm-Liouville problems [i.e., equations \eqref{EigenDiff}\textendash\eqref{EigenB2} with $c_i=d_i=0$ for $i=1,2$] are eigenvalue problems in some subspace of $L^2$ \citep{debnath_introduction_2005}. For the more general case of interest here, the eigenvalue problem occurs over a ``larger'' function space denoted by $L^2_\mu$ which we construct in appendix \ref{S-additional-properties}.

Let the integer $s\in\{0,1,2\}$ denote the number of $\lambda$-dependent boundary conditions and let $S$ denote the set
\begin{equation}\label{S-set}
	S = \{j \ | \ j \in \{1,2\} \text{ and } (c_j,d_j) \neq (0,0) \}.
\end{equation}
$S$ is one of $\emptyset,\{1\},\{2\},\{1,2\}$ and $s$ is the number of elements in the set $S$. In appendix \ref{S-additional-properties}, we show that $L^2_\mu$ is isomorphic to the space $L^2\oplus \C^s$ and is thus ``larger'' than $L^2$ by $s$ dimensions. 

We denote elements of $L^2_\mu$ by upper case letters $\Psi$; we define $\Psi(z)$ for $z\in[z_1,z_2]$ by
\begin{equation}\label{L2_mu-element}
	\Psi(z) = 
	\begin{cases}
		\Psi(z_i) \quad & \textrm{at } z=z_i, \textrm{ for }\, i \in S, \\
		\psi(z) \quad & \textrm{otherwise},
	\end{cases}
\end{equation}
where $\Psi(z_i) \in \C$ are constants, for $i\in S$, and the corresponding lower case letter $\psi$ denotes an element of $L^2$. 
Two elements $\Phi$ and $\Psi$ of $L^2_\mu$ are equivalent in $L^2_\mu$ if and only if
\begin{itemize}
	\item [\emph{1.}] $\Phi(z_i) = \Psi(z_i)$ for $i\in S$, and
	\item [\emph{2.}] $\phi(z)$ and $\psi(z)$ are equivalent in $L^2$ [i.e., as in equation \eqref{L2-equal}].
\end{itemize}
Here, $\Phi$, as an element of $L^2_\mu$, is defined as in equation \eqref{L2_mu-element}. The primary difference between $L^2$ and $L^2_\mu$ is that $L^2_\mu$ discriminates between functions that disagree at $\lambda$-dependent boundaries.

The measure $\mu$ associated with $L^2_\mu$ (see appendix \ref{S-additional-properties}) induces a differential element
\begin{equation} \label{dmu}
	\mathrm{d}\mu(z) = \left[ r(z) + \sum_{i\in S} D_i^{-1} \, \delta(z-z_i) \right] \mathrm{d}z,
\end{equation}
where $\delta(z)$ is the Dirac delta. The induced inner product on $L^2_\mu$ is
\begin{equation}\label{dmu_inner}
	\inner{\Phi}{\Psi} = \intz \ov{\Phi} \, {\Psi} \, \mathrm{d}\mu = \intz \ov {\Phi} \, {\Psi} \, r \, \mathrm{d}z + \sum_{i\in S} D_i^{-1} \, \ov{\Phi(z_i)} \, \Psi(z_i).
\end{equation}
If $D_i > 0$ for $i\in S$ then this inner product is positive-definite and $L^2_\mu$ is a Hilbert space. However, this is not the case in general.

Let $\kappa$ denote the number of negative $D_i$ for $i\in S$ (the possible values are $\kappa=0,1,2$). Then $L^2_\mu$ has a $\kappa$-dimensional subspace of elements $\Psi$ satisfying 
\begin{equation}
	\inner{\Psi}{\Psi} < 0.
\end{equation}
This makes $L^2_\mu$ a Pontryagin space of index $\kappa$ \citep{bognar_indefinite_1974}. If $\kappa=0$ then $L^2_\mu$ is again a Hilbert space. In the present case, $L^2_\mu$ also has an infinite-dimensional subspace of elements $\psi$ satisfying 
\begin{equation}
	\inner{\Psi}{\Psi} > 0.
\end{equation}

\subsubsection{Reality and completeness} \label{S-real-complete}

In appendix \ref{S-eigen-in-L2-mu}, we reformulate the eigenvalue problem \eqref{EigenDiff}\textendash\eqref{EigenB2} as an eigenvalue problem of the form
\begin{equation}\label{Op-Eigenproblem}
	\Lc \, \Phi = \lambda \, \Phi
\end{equation}
in a subspace of $L^2_\mu$, where $\Lc$ is a linear operator and $\Phi$ an element of $L^2_\mu$. We also define the notions of right- and left-definiteness that are required for the reality and completeness theorem below. The following two propositions can be considered to define right- and left-definiteness for applications of the theory. Both propositions are obtained through straightforward manipulations (see appendix \ref{S-additional-properties}).

\begin{proposition}[Criterion for right-definiteness]\label{right-definite-prop}
	The eigenvalue problem \eqref{EigenDiff}\textendash\eqref{EigenB2} is right-definite if $r>0$ and $D_i>0$ for $i\in S$.
\end{proposition}

\begin{proposition}[Criterion for left-definiteness] \label{left-definite-prop}
	The eigenvalue problem \eqref{EigenDiff}\textendash\eqref{EigenB2} is left-definite if the following conditions hold:
	\begin{itemize}
		\item [(i)] the functions $p,q$ satisfy $p>0, q \geq 0$,
		\item [(ii)] for the $\lambda$-dependent boundary conditions, we have 
			\begin{equation}
				\frac{a_i \, c_i}{D_i} \leq 0, \quad \frac{b_i \, d_i}{D_i} \leq 0, \quad (-1)^i \frac{a_i \, d_i}{D_i} \geq 0  \quad \text{for } i \in S.
			\end{equation}
		\item [(iii)] for the $\lambda$-independent boundary conditions, we have 
			\begin{equation}
				b_i = 0  \quad \text{or} \quad  (-1)^{i+1}\frac{a_i}{b_i} \geq 0 \quad \text{ if } b_i \neq 0   \quad \text{ for } i \in \{1,2\}\setminus S.
			\end{equation}
	\end{itemize}
\end{proposition}
The notions of right and left-definiteness are not mutually exclusive. Namely, a problem can be neither right- or left-definite; both right- and left-definite; only right-definite; or only left-definite. In this article, we always assume that $p>0$ and $r>0$.

%%% NEW ADD
%All of the examples we consider later are left-definite. In non-canonical Hamiltonian problems, left-definiteness seems to be related to the formal stability condition in \cite{holm_nonlinear_1985}. 

The reality of the eigenvalues and the completeness of the eigenfunctions in the space $L^2_\mu$ is given by the following theorem.
\begin{theorem}[Reality and completeness]\label{real-complete}

	Suppose the eigenvalue problem \eqref{EigenDiff}\textendash\eqref{EigenB2} is either right-definite or left-definite. Moreover, if the problem is not right-definite, we assume that $\lambda=0$ is not an eigenvalue. Then the eigenvalue problem  \eqref{EigenDiff}\textendash\eqref{EigenB2} has a countable infinity of real simple eigenvalues $\lambda_n$ satisfying 
	\begin{equation}
		\lambda_0 < \lambda_1 < \dots < \lambda_n < \dots \rightarrow \infty,
	\end{equation}
	with corresponding eigenfunctions $\Phi_n$. Furthermore, the set of eigenfunctions $\{\Phi_n\}_{n=0}^\infty$ is a complete orthonormal basis for $L^2_\mu$ satisfying
	\begin{equation}
		\inner{\Phi_m}{\Phi_n} = \pm \delta_{mn}.
	\end{equation} 
\end{theorem}
\begin{proof}
	See appendix \ref{S-real-proof}.
\end{proof}
Recall that $\kappa$ denotes the number of negative $D_i$ for $i\in S$. We then have the following corollary of the proof of theorem \ref{real-complete}.

\begin{proposition}\label{eigenvalue-sign}
	Suppose the eigenvalue problem \eqref{EigenDiff}\textendash\eqref{EigenB2} is left-definite and that $\lambda=0$ is not an eigenvalue. Then there are $\kappa$ negative eigenvalues and their eigenfunctions satisfy
	\begin{equation}
		\inner{\Phi}{\Phi} < 0.
	\end{equation}
	The remaining eigenvalues are positive and their eigenfunctions satisfy
	\begin{equation}
		\inner{\Phi}{\Phi}>0.
	\end{equation}
\end{proposition}
In other words, proposition \ref{eigenvalue-sign} states that we have the relationship
\begin{equation}
	\lambda_n \inner{\Phi_n}{\Phi_n} > 0
\end{equation}
for left-definite problems.

\subsection{Properties of the eigenfunctions} \label{S-Prop-of-Eig}
For the remainder of \S \ref{math-section}, we assume that the eigenvalue problem \eqref{EigenDiff}\textendash\eqref{EigenB2} satisfies the requirements of theorem \ref{real-complete}.

\subsubsection{Eigenfunction expansions}

The eigenvalue problem \eqref{EigenDiff}\textendash\eqref{EigenB2} has \emph{eigenfunctions} $\{\Phi_n\}_{n=0}^\infty$ as well as corresponding \emph{solutions} $\{\phi_n\}_{n=0}^\infty$. In other words, while the $\phi_n$ are the solutions to the differential equation defined by equations \eqref{EigenDiff}\textendash\eqref{EigenB2} with $\lambda=\lambda_n$, the eigenfunctions required by the operator formulation of the problem [equation \eqref{Op-Eigenproblem}] are $\Phi_n$. The functions $\Phi_n$ and $\phi_n$ are related by equation \eqref{L2_mu-element}, with the boundary values $\Phi_n(z_i)$ of $\Phi_n$ determined by
\begin{equation}\label{discontinuity}
	\Phi_n(z_i) = \left[c_i \, \phi(z) - d_i \, (p \, \phi')(z)\right] \quad \text{for } i\in S.
\end{equation}
Thus, while the solutions $\phi_n$ are continuously differentiable over the closed interval $[z_1,z_2]$, the eigenfunctions $\Phi_n$ are continuously differentiable over the open interval $(z_1,z_2)$ but generally have finite jump discontinuities at the $\lambda$-dependent boundaries.
The eigenfunctions $\Phi_n$ are continuous in the closed interval $[z_1,z_2]$ only if $c_i=1$ and $d_i=0$ for $i\in S$. In this case, the eigenfunctions $\Phi_n$ coincide with the solutions $\phi_n$ on the closed interval $[z_1,z_2]$.

The boundary conditions of the eigenvalue problem \eqref{EigenDiff}\textendash\eqref{EigenB2} are not unique. One can multiply each boundary condition by an arbitrary constant to obtain an equivalent problem. To uniquely specify the eigenfunctions in physical applications, the boundary coefficients $\{a_i,b_i,c_i,d_i\}$ of equations \eqref{EigenDiff}\textendash\eqref{EigenB2} must be chosen so that $r \, \mathrm{d}z$ has the same dimensions as $D_i^{-1} \, \delta(z-z_i) \, \mathrm{d}z$ [recall that $\delta(z)$ has the dimension of inverse length]. In the quasigeostrophic problem, we must also invoke continuity and set $c_i=1$.

Since $\{\Phi_n\}_{n=0}^\infty$ is a basis for $L^2_\mu$, then any $\Psi \in L^2_\mu$ may be expanded in terms of the eigenfunctions \citep[][thereom IV.3.4]{bognar_indefinite_1974},
\begin{equation}\label{expansion}
	\Psi = \sum_{n=0}^\infty \frac{\inner{\Psi}{\Phi_n}}{\inner{\Phi_n}{\Phi_n}} \, \Phi_n.
\end{equation}
We emphasize that the above equality is an equality in $L^2_\mu$ and not a pointwise equality [see the discussion following equation \eqref{L2_mu-element}]. Some properties of $L^2_\mu$ expansions are given in appendix \ref{S-more-eigen-expansions}.

An important property that distinguishes the basis $\{\Phi_n\}_{n=0}^\infty$ of $L^2_\mu$ from an $L^2$ basis is its ``sensitivity'' to function values at boundary points $z=z_i$ for $i\in S$. See \S \ref{S-step-forcing} for a physical application.

A natural question is whether the basis $\{\Phi_n\}_{n=0}^\infty$ of $L^2_\mu$ is also a basis of $L^2$. Recall that the set $\{\Phi_n\}_{n=0}^\infty$ is a basis of $L^2$ if every element $\psi \in L^2$ can be written \emph{uniquely} in terms of the functions $\{\Phi_n\}_{n=0}^\infty$. However, in general, this is not true. If $s>0$, the $L^2_\mu$ basis $\{\Phi_n\}_{n=0}^\infty$ is overcomplete in $L^2$ \citep{walter_regular_1973,russakovskii_matrix_1997}.

\subsubsection{Uniform convergence and term-by-term differentiability}\label{S-uniform}

Along with the eigenfunction expansion \eqref{expansion} in terms of the eigenfunctions $\{\Phi_n\}_{n=0}^\infty$, we also have the expansion
\begin{equation}\label{expansion-phi}
	\sum_{n=0}^\infty \frac{\inner{\Psi}{\Phi_n}}{\inner{\Phi_n}{\Phi_n}} \, \phi_n
\end{equation}
in terms of the solutions $\phi_n$. The two expansions differ in their behaviour at $\lambda$-dependent boundaries, $z=z_i$ for $i\in S$, but are otherwise equal. In particular, the $\Phi_n$ eigenfunction expansion \eqref{expansion} must converge to $\Psi(z_i)$ at $z=z_i$ for $i\in S$ as this equality is required for $\Psi$ to be equal to the series expansion \eqref{expansion} in $L^2_\mu$ [see the discussion following equation \eqref{L2_mu-element}]. Some properties of both expansions are given in appendix \ref{S-pointwise}. In particular, theorem \ref{pointwise} shows that the $\phi_n$ solution series \eqref{expansion-phi} does not generally converge to $\Psi(z_i)$ at $z=z_i$.

\begin{figure}
  \includegraphics[width=1.\columnwidth]{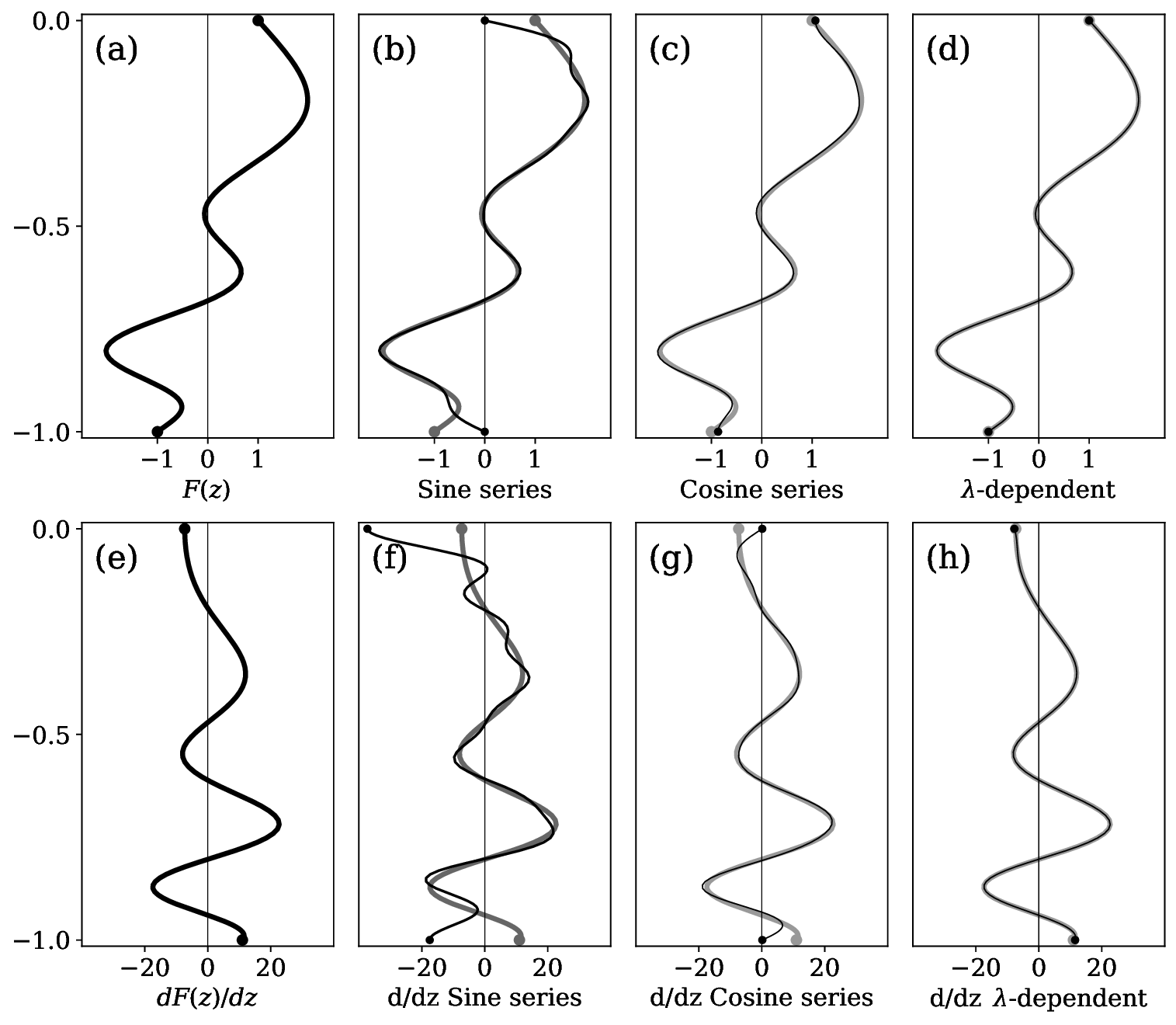}
  \caption{Convergence to a function $F(z) = 1 + 2z +(3/2)\sin(2\pi z)\cos(\pi^2z^2 + 3)$ for $z \in[-1,0]$, shown in panel (a), by various eigenfunction expansions of $-\phi''= \lambda \, \phi$ with fifteen terms, as discussed in \S \ref{S-uniform}. Panel (b) shows the Fourier sine expansion of $F$. Since the sine eigenfunctions vanish at the boundaries $z=-1,0$, the series expansion will not converge to $F$ at the boundaries. Panel (c) shows the cosine expansion of $F$ which converges uniformly to $F$ on the closed interval $[-1,0]$. Panel (d) shows an expansion with boundary coefficients in equations \eqref{EigenB1}\textendash\eqref{EigenB2} given by $(a_1,b_1,c_1,d_1) = (-0.5,-5,1,0)$ and $(a_2,b_2,c_2,d_2) =  (0.5,-5,1,0)$. Since the $c_i=1$ and $d_i=0$, then $\Phi_n = \phi_n$ and the series expansions \eqref{expansion} and \eqref{expansion-phi} coincide. As with the cosine series, the expansion converges uniformly to $F$ on $[-1,0]$. The derivative of $F$ is shown in panel (e). Panel (f) show the derivative of the sine series expansion. In panel (g), we show the differentiated cosine series which does not converge to the derivative $F'$ at the boundaries $z=z_1,z_2$. In contrast, in panel (h), the differentiated series obtained from a problem with $\lambda$-dependent boundary conditions converges uniformly to the derivative $F'$.}
  \label{F-convergence}
\end{figure}

The following theorem is of central concern for physical applications. 

\begin{theorem}[Uniform convergence]\label{uniform}
 Let $\psi$ be a twice continuously differentiable function on $[z_1,z_2]$ satisfying all $\lambda$-independent boundary conditions of the eigenvalue problem \eqref{EigenDiff}\textendash\eqref{EigenB2}. Define the function $\Psi$ on $z\in[z_1,z_2]$ by 
	\begin{equation}
	\Psi(z) = 
	\begin{cases}
		c_i \, \psi(z) - d_i \, (p \, \psi')(z) \quad & \textrm{at } z=z_i, \textrm{ for }\, i \in S, \\
		\psi(z) \quad & \textrm{otherwise}.
	\end{cases}
	\end{equation}
	Then 
	\begin{equation}\label{uniform-two-series}
		\psi(z) = \sum_{n=0}^\infty \frac{\inner{\Psi}{\Phi_n}}{\inner{\Phi_n}{\Phi_n}} \, \phi_n(z)  \quad \textrm{and} \quad \psi'(z) = \sum_{n=0}^\infty \frac{\inner{\Psi}{\Phi_n}}{\inner{\Phi_n}{\Phi_n}} \, \phi_n'(z) 
	\end{equation}
	with both series converging uniformly and absolutely on $[z_1,z_2]$.
\end{theorem}
\begin{proof}
	See appendix \ref{S-A-left-fulton}.
\end{proof}

If $c_i=1$ and $d_i=0$ for $i\in S$ then we can replace $\Phi_n$ by $\phi_n$ and $\Psi$ by $\psi$ in equation \eqref{uniform-two-series}.

In addition, if both boundary conditions of the eigenvalue problem \eqref{EigenDiff}\textendash\eqref{EigenB2} are $\lambda$-dependent, then both  expansions in equation \eqref{uniform-two-series} converge uniformly on $[z_1,z_2]$ regardless of the boundary conditions $\psi$ satisfies. As discussed in appendix \ref{S-pointwise}, for traditional Sturm-Liouville expansions, an analogous result holds only if $\psi$ satisfies the same boundary conditions as the eigenfunctions. Figure \ref{F-convergence} contrasts the convergence behaviour of such a problem (with continuous eigenfunctions, so $c_i=1$ and $d_i=0$ for $i\in S$) with the convergence behaviour of sine and cosine series. All numerical solutions in this article are obtained using a pseudo-spectral code in Dedalus \citep{burns_dedalus_2020}. 

\begin{figure}
  \centerline{\includegraphics[width=0.8\columnwidth]{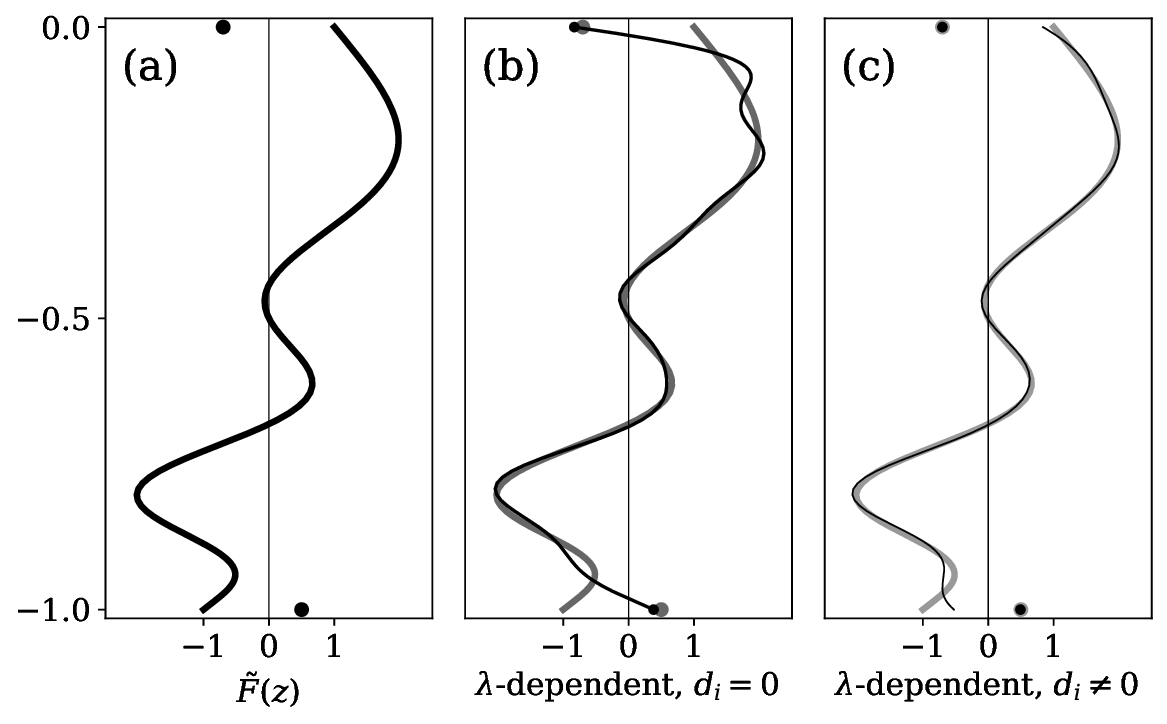}}
  \caption{Convergence to a function $\tilde F$ with finite jump discontinuities at the boundaries by two eigenfunction expansions (with $\lambda$-dependent boundary conditions) of $-\phi'' = \lambda\,\phi$ with fifteen terms, as discussed in \S \ref{S-uniform}. The function $\tilde F(z)$ is defined by $\tilde F(z) = F(z)$ for $z\in(z_1,z_2)$ where $F(z)$ is the function defined in figure \ref{F-convergence}, $F(-1) = 0.5$ at the lower boundary, and $F(0)=-0.7$ at the upper boundary. The function $\tilde F$ is shown in panel (a). In panel (b), the boundary coefficients in equations \eqref{EigenB1}\textendash\eqref{EigenB2} are given by  $(a_1,b_1,c_1,d_1) = (-0.5,-5,1,0)$ and $(a_2,b_2,c_2,d_2) =  (0.5,-5,1,0)$ as in figure \ref{F-convergence}. In panel (c), the boundary coefficients are  $(a_1,b_1,c_1,d_1) = (-0.5,-5,1,0.1)$ and $(a_2,b_2,c_2,d_2) =  (0.5,-5,1,-0.1)$. The $\Phi_n$ expansion \eqref{expansion} and the $\phi_n$ expansion \eqref{expansion-phi} are not generally equal at the boundaries $z=-1,0$; this figure shows the $\Phi_n$ expansion. The $\Phi_n$ series \eqref{expansion} converges pointwise to $\tilde F$ on $[-1,0]$, however, the convergence will not be uniform if $d_i=0$ for $i\in S$, as in panel (b). The boundary values of the $\Phi_n$ series \eqref{expansion} are shown with a black dot. In panel (b), the eigenfunctions $\Phi_n$ are continuous and a large number of terms are required for the series to converge to the discontinuous function $\tilde F$. Panel (c) shows that the discontinuous eigenfunction $\Phi_n$ have almost converged to the $\tilde F$\textemdash including at the jump discontinuities; the black dot in panel (c) overlap with the grey dots, which represent the boundary values of $\tilde F$. Although the $\phi_n$ series \eqref{expansion-phi} converges to $\tilde F$ in the interior $(-1,0)$, the $\phi_n$ series does not generally converge to $\tilde F$ at the boundaries but instead converges to the values given in theorem \ref{pointwise}.} 
  \label{F-convergence-disc}
\end{figure}

Another novel property of the eigenfunction expansions is that we obtain pointwise convergence to functions that are smooth in the interior of the interval, $(z_1,z_2)$, but have finite jump discontinuities at $\lambda$-dependent boundaries (see appendix \ref{S-pointwise}). If $d_i\neq 0$ for $i\in S$, the convergence is even uniform \cite[][corollary 2.1]{fulton_two-point_1977}. Figure \ref{F-convergence-disc} illustrates the convergence behaviour for eigenfunction expansions with $\lambda$-dependent boundary conditions in the two cases $d_i=0$ and $d_i\neq 0$. Note the presence of Gibbs-like oscillations in the case $d_i=0$ shown in panel (b). Although the $\Phi_n$ eigenfunction series \eqref{expansion} converges pointwise to the discontinuous function, the $\phi_n$ solution series \eqref{expansion-phi} converges to the values given in theorem \ref{pointwise} at the $\lambda$-dependent boundaries. The ability of these series expansions to converge to functions with boundary jump discontinuities is related to their ability to expand distributions in the \citet{bretherton_critical_1966} ``$\delta$-function formulation'' of a problem.

\subsection{Oscillation theory}\label{S-oscillation}

Recall that for regular Sturm-Liouville problems [i.e., equations \eqref{EigenDiff}\textendash\eqref{EigenB2} with $c_i=d_i=0$] we obtain a countable infinity of real simple eigenvalues, $\lambda_n$, that may be ordered as 
\begin{equation}\label{eig-seq}
	\lambda_0 < \lambda_1 < \lambda_2 < \dots \rightarrow \infty,
\end{equation}
with associated eigenfunctions $\phi_n$. The $n$th eigenfunction $\phi_n$ has $n$ internal zeros in the interval $(z_1,z_2)$ so that no two eigenfunctions have the same number of internal zeros.

However, once the eigenvalue $\lambda$ appears in the boundary conditions, there may be up to $s+1$ linearly independent eigenfunctions with the same number of internal zeros. The crucial parameters deciding the number of zeros is $-b_i/d_i$ for $i\in S$, where $b_i$ and $d_i$ are the boundary coefficients appearing in the boundary conditions \eqref{EigenB1}\textendash\eqref{EigenB2}. The following lemma outlines the possibilities when only one boundary condition is $\lambda$-dependent.

\begin{lemma}[Location of double oscillation count]\label{extra-oscillation}

Suppose that $s=1$, $i\in S$, and let $\kappa$ be the number of negative $D_i$ for the eigenvalue problem \eqref{EigenDiff}\textendash\eqref{EigenB2}. We have the following possibilities.
\begin{itemize}
	\item[(i)] Right-definite, $d_i \neq 0$: The eigenfunction $\Phi_n$ corresponding to the eigenvalue $\lambda_n$ has $n$ internal zeros if $\lambda_n < -b_i/d_i$ and $n-1$ internal zero if $-b_i/d_i \leq \lambda_n$.
	\item[(ii)] Right-definite, $d_i = 0$:  The $n$th eigenfunction has $n$ internal zeros.
	\item[(iii)] Left-definite: If $\kappa=0$ then all eigenvalues are positive, the problem is right-definite, and  either (i) or (ii) applies.
	Otherwise, if $\kappa=1$, then the eigenvalues may be ordered as 
	\begin{equation}
		\lambda_0 < 0 < \lambda_1 < \lambda_2 < \dots  \rightarrow \infty.
	\end{equation}
	Both eigenfunctions $\Phi_0$ and $\Phi_1$ have no internal zeros. The remaining eigenfunctions $\Phi_n$, for $n>1$, have $n-1$ internal zeros.
\end{itemize}
\end{lemma}
\begin{proof}
	Parts (i), (ii) and (ii) are due to \cite{linden_leightons_1991}, \cite{binding_sturmliouville_1994}, and \cite{binding_left_1999}, respectively.
\end{proof}

When both boundary conditions are $\lambda$-dependent, the situation is similar. See \cite{binding_sturmliouville_1994} and \cite{binding_left_1999} for further discussion.

\section{Boussinesq gravity-capillary waves} \label{S-Boussinesq}

Consider a rotating Boussinesq fluid on an $f$-plane with a reference Boussinesq density of $\rho_0$. The fluid is subject to a constant gravitational acceleration $g$ in the downwards, $-\unit z$, direction, and to a surface tension $T$ \citep[with dimensions of force per unit length, see][]{lamb_hydrodynamics_1975} at its upper boundary. The upper boundary of the fluid, given by $z=\eta$, is a free-surface defined by the function $\eta(\vec x,t)$, where $\vec x = \unit x \, x + \unit y \, y$ is the horizontal position vector. The lower boundary of the fluid is a flat rigid surface given by $z=-H$. The fluid region is periodic in both horizontal directions $\unit x$ and $\unit y$. 

\subsection{Linear equations of motion}

The governing equations for infinitesimal perturbations about a background state of no motion, characterized by a prescribed background density of $\rho_B = \rho_B(z)$, are 
	\begin{align} \label{grav-cap-1}
	\partial_t^2  \lap  w + f_0^2 \, \partial_z^2  w + N^2 \, \lap_z w = 0 \quad &\text{for } z \in \left(-H,0\right) \\
	\label{grav-cap-2}
	w = 0 \quad &\text{for } z=-H \\
	\label{grav-cap-3}
	-\partial_t^2 \partial_z w - f_0^2 \, \partial_z w + g_b \, \lap_z  w - \tau \, \laptwo_z  w = 0 \quad &\text{for } z=0,
\end{align}
where $w$ is the vertical velocity, $f_0$ is the constant value of the Coriolis frequency, the prescribed buoyancy frequency $N^2$ is given by
\begin{equation}\label{N2}
	N^2(z) = - \frac{g}{\rho_0} \d{\rho_B(z)}{z},
\end{equation}
the acceleration $g_b$ is the effective gravitational acceleration at the upper boundary
\begin{equation}
	g_b = -\frac{g}{\rho_0} \left[\rho_a - \rho_B(0-)\right]
\end{equation}
where $\rho_a$ is the density of the overlying fluid, and the parameter $\tau$ is given by
\begin{equation}
	\tau = \frac{T}{\rho_0}
\end{equation}
where $T$ is the surface tension. The three-dimensional Laplacian is denoted $\lap = \partial_x^2 + \partial_y^2 + \partial_z^2$, the horizontal Laplacian is denoted by $\lap_z = \partial_x^2 + \partial_y^2$, and the horizontal biharmonic operator is given by $\laptwo_z = \lap_z \, \lap_z$. See equation (1.37) in \cite{dingemans_water_1997} for the surface tension term in \eqref{grav-cap-3}. The remaining terms in equation \eqref{grav-cap-1}\textendash\eqref{grav-cap-3} are standard \citep{gill_atmosphere-ocean_2003}. Consistent with our assumption that $\eta(\vec x,t)$ is small, we evaluate the upper boundary condition at $z=0$ in equation \eqref{grav-cap-3}.
 
\subsection{Non-rotating Boussinesq fluid}\label{S-Boussinesq-non-rot} 

\begin{figure}
  \centerline{\includegraphics[width=1.\columnwidth]{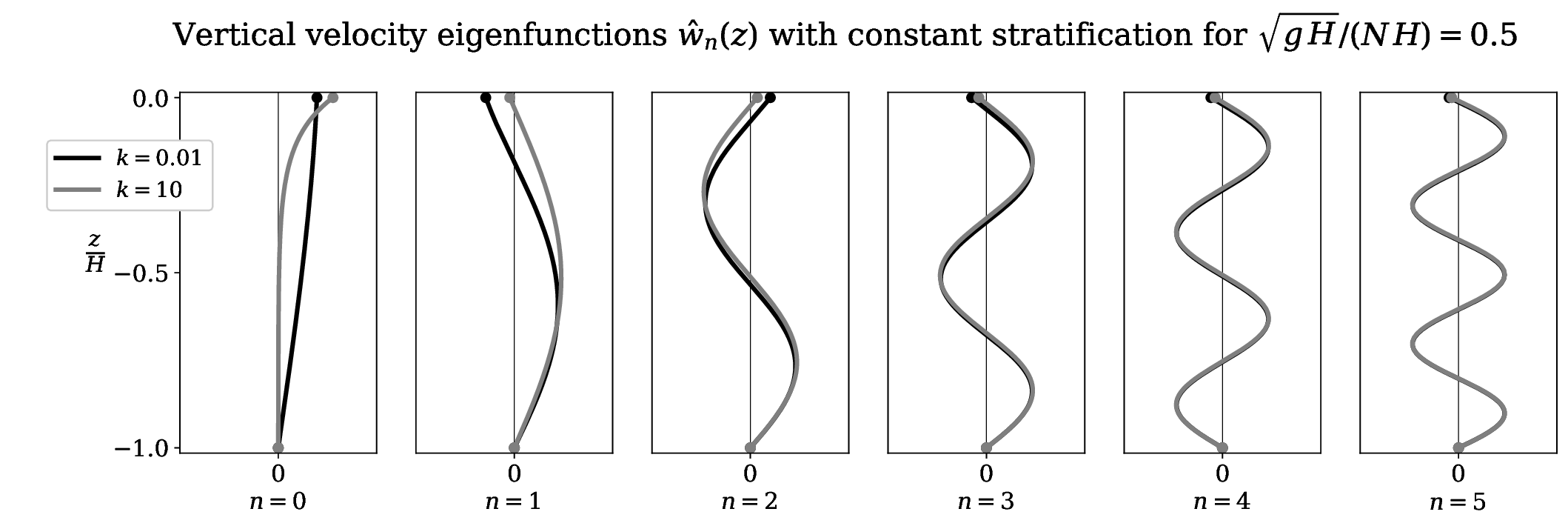}}
  \caption{The vertical velocity eigenfunctions $\hat W_n = \hat w_n$ of the non-rotating Boussinesq eigenvalue problem \eqref{non-rot-SL1}\textendash\eqref{non-rot-SL3} for two distinct wavenumbers with constant stratification, as discussed in \S \ref{S-Boussinesq-non-rot}. For both wavenumbers, the $n$th eigenfunction has $n$ internal zeros as in regular Sturm-Liouville theory. The zeroth mode ($n=0$) corresponds to a surface gravity wave and is trapped to the upper boundary for large horizontal wavenumbers. In contrast to the internal wave problem with a rigid lid, the modes $\hat w_n$ now depend on the horizontal wavenumber $k$ through the boundary condition \eqref{non-rot-SL3}, however, this dependence is weak for $n\gg1$, as can be observed in this figure; for $n>2$, the modes for $k=0.01$ (in black)  and for $k=10$ (in grey) nearly coincide. The horizontal wavenumbers $k$ are non-dimensionalized by $H$.}
  \label{F-nonrot-eigenfunctions}
\end{figure}

We assume wave solutions of the form
\begin{equation}\label{w-wave}
	w(\vec x, z ,t) = \hat w(z) \, \mathrm{e}^{\mathrm{i} \left( \vec k \cdot \vec x - \omega t\right)}
\end{equation}
where $\vec k = \unit x \, k_x + \unit y \, k_y$ is the horizontal wavevector and $\omega$ is the angular frequency. Substituting the wave solution \eqref{w-wave} into equations \eqref{grav-cap-1}\textendash\eqref{grav-cap-3} and setting $f_0=0$ yields
\begin{align}
	\label{non-rot-SL1}
	-\hat w'' + k^2 \, \hat w = \sigma^{-2} \, N^2 \, \hat w \quad &\text{for } z \in (-H,0) \\
	\label{non-rot-SL2}
	\hat w = 0 \quad &\text{for } z = -H \\
	\label{non-rot-SL3}
	(g_b + \tau \, k^2)^{-1} \hat w ' =  \sigma^{-2} \, \hat w \quad &\text{for } z = 0,
\end{align}
where $\sigma = \omega/k$ is the phase speed and $k = \abs{\vec k}$ is the horizontal wavenumber. Equations \eqref{non-rot-SL1}\textendash\eqref{non-rot-SL3} are an eigenvalue problem for the eigenvalue $\lambda = \sigma^{-2}$.

\subsubsection*{Definiteness \& the underlying function space}

Equations \eqref{non-rot-SL1}\textendash\eqref{non-rot-SL3} form an eigenvalue problem with one $\lambda$-dependent boundary condition, namely, the upper boundary condition \eqref{non-rot-SL3}. The underlying function space is then
\begin{equation} \label{gravity-functionspace}
	L^2_\mu \cong L^2 \oplus \C.
\end{equation}
We write $\hat W_n$ for the eigenfunctions and $\hat w_n$ for the solutions of the eigenvalue problem \eqref{non-rot-SL1}\textendash\eqref{non-rot-SL3} [see the paragraph containing equation \eqref{discontinuity}]. The eigenfunctions $\hat W_n$ are related to the solutions $\hat w_n$ by equation \eqref{L2_mu-element} with boundary values $\hat W_n(0)$ given by equation \eqref{discontinuity}. However, since $c_2 = 1$ and $d_2=0$ in equation \eqref{non-rot-SL3} [compare with equations \eqref{EigenDiff}\textendash\eqref{EigenB2}] then $\hat W_n = \hat w_n$ on the closed interval $[-H,0]$; thus, the solutions $w_n$ are also the eigenfunctions.

By theorem \ref{real-complete}, the eigenfunctions $\{\hat w_n \}_{n=0}^\infty$ form an  orthonormal basis of $L^2_\mu$. For functions $\varphi$ and $\phi$, the inner product is
\begin{align}\label{grav-inner}
	\inner{\varphi}{\phi} &= \frac{1}{N_0^2 \, H} \left[ \int_{-H}^0 \varphi \, \phi \, N^2 \, \textrm{d}z + (g_b + \tau \, k^2) \varphi(0) \, \phi(0)\right]
\end{align}
obtained from equations \eqref{dmu_inner} and equation \eqref{Di}; we have introduced the factor $1/(N^2_0\, H)$ in the above expression for dimensional consistency in eigenfunction expansions ($N^2_0$ is a typical value of $N^2$). Orthonormality is then given by 
\begin{align}\label{non-rot-ortho}
	\delta_{mn} &= \inner{\hat w_m}{\hat w_n}
\end{align}
and we have chosen the solutions $\hat w_n$ to be non-dimensional (so the Kronecker delta is non-dimensional as well).

One verifies that the eigenvalue problem \eqref{non-rot-SL1}\textendash\eqref{non-rot-SL3} is right-definite using proposition \ref{right-definite-prop} and left-definite using proposition \ref{left-definite-prop}.
Right-definiteness implies that $L^2_\mu$, with the inner product \eqref{grav-inner}, is a Hilbert space. That is, all eigenfunctions $\hat w_n$ satisfy
\begin{equation}
	\inner{\hat w_n}{\hat w_n} > 0.
\end{equation} 
Left-definiteness, along with proposition \ref{eigenvalue-sign}, ensures that all eigenvalues $\lambda_n = \sigma_n^{-2}$ are positive. Indeed, the phase speeds $\sigma_n$ satisfy
\begin{equation}\label{grav-phase-speeds}
	\sigma_0^2 > \sigma_1^2 > \dots > \sigma^2_n > \dots  \rightarrow 0.
\end{equation}

\subsubsection*{Properties of the eigenfunctions}

By lemma \ref{extra-oscillation}, the $n$th eigenfunction $\hat w_n$ has $n$ internal zeros in the interval $(-H,0)$. See figure \ref{F-nonrot-eigenfunctions} for an illustration of the first six eigenfunctions.

The eigenfunctions $\{\hat w_n \}_{n=0}^\infty$ are complete in $L^2$ but do not form a basis in $L^2$; in fact, the basis is overcomplete in $L^2$. The presence of a free-surface provides an additional degree of freedom over the usual  rigid-lid $L^2$ basis of internal wave eigenfunctions. Indeed, the $n=0$ wave in figure \ref{F-nonrot-eigenfunctions} corresponds to a surface gravity wave, while the remaining modes are internal gravity waves (with some surface motion).

\subsubsection*{Expansion properties}

Given a twice continuously differentiable function $\chi (z)$ satisfying $\chi(-H)=0$, then, from theorem \ref{uniform}, we have
\begin{equation}
	\chi(z) = \sum_{n=0}^\infty \inner{\chi}{\hat w_n} \, \hat w_n(z) \quad \textrm{and} \quad \chi'(z) = \sum_{n=0}^\infty \inner{\chi}{\hat w_n}\, w_n'(z),
\end{equation}
with both series converging uniformly on $[-H,0]$ (note that $\chi$ is not required to satisfy any particular boundary condition at $z=0$). If $\chi$ is the vertical structure at time $t=0$ (and at some wavevector $\vec k$) and we assume $\partial_t w(\vec x,z,t=0)=0$, then the subsequent time-evolution is given by
\begin{equation}
	 w(\vec x, z,t) = \sum_{n=0}^\infty \inner{\chi}{\hat w_n} \, w_n(z) \, \cos\left(\sigma_n k t\right) \mathrm{e}^{\mathrm{i}\vec k \cdot \vec x} .
\end{equation}

%%% NEW ADD
\subsubsection*{The $f$-plane hydrostatic problem}

Suppose we have hydrostatic gravity waves on an $f$-plane with free surface at the upper boundary, as in \cite{kelly_vertical_2016}. The appropriate inner product is obtained by setting $\tau = 0$ in the inner product \eqref{grav-inner}. All the above results on the eigenfunctions of gravity-capillary waves carry over to the hydrostatic $f$-plane problem provided we set 
\begin{equation}
    \sigma^2 = \frac{\omega^2 - f_0^2}{k^2}.
\end{equation}

\subsection{A Boussinesq fluid with a rotating upper boundary} \label{S-Boussinesq-rot}

\begin{figure}
  \centerline{\includegraphics[width=1.\columnwidth]{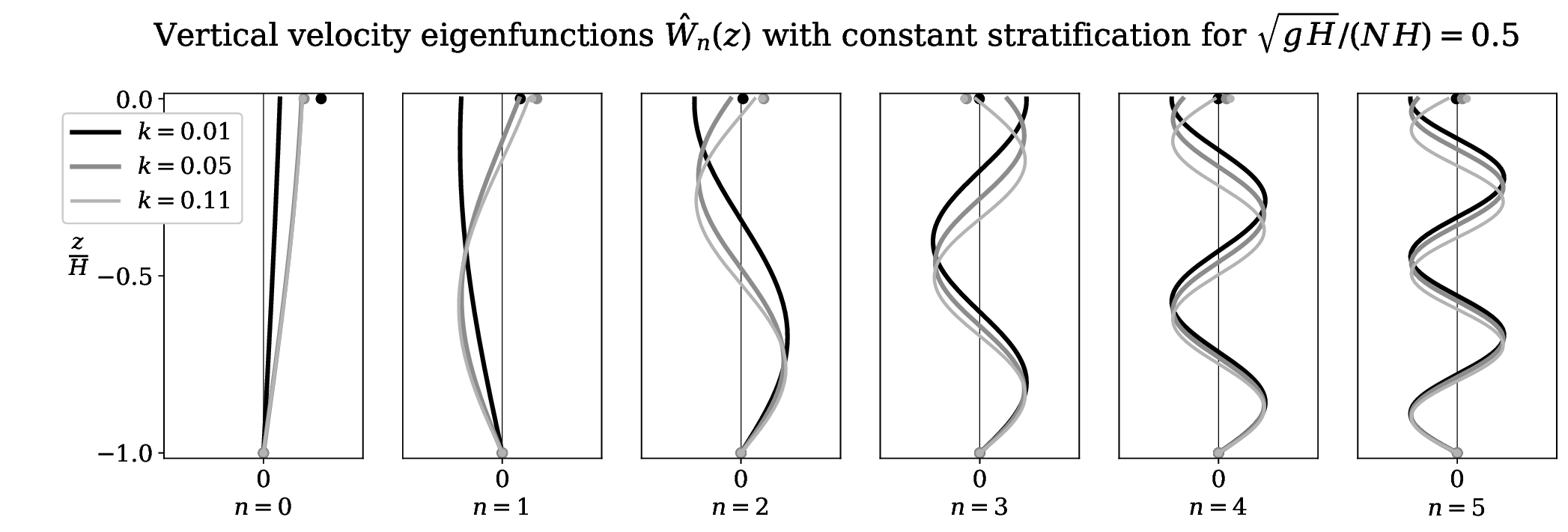}}
  \caption{The vertical velocity eigenfunctions $\hat W_n$ of a Boussinesq fluid with a rotating upper boundary\textemdash eigenvalue problem \eqref{rot-bound-SL1}\textendash\eqref{rot-bound-SL3}. This figure is discussed in \S \ref{S-Boussinesq-rot}. The wavenumbers $k$ in the figure are non-dimensionalized by the depth $H$. The dots represent the values of the eigenfunctions at the boundaries. Note that the eigenfunctions have a finite jump discontinuity at $z=0$. For $k \, H=0.01$ (given by the black line) there are two modes with no internal zeros. As $k$ increases, we obtain two modes with one internal zero (at $k\, H =0.05$, the thick grey line) and then two modes with three internal zeros (at $k\, H=0.11$, the thin grey line).}
  \label{F-rotating-eigenfunctions}
\end{figure}
%%% NEW ADD
Although this next example is not geophysically relevant, it has the curious property that the resulting eigenfunctions are discontinuous.

Let $N^2_0$ be a typical value of $N^2(z)$. Consider the situation where $f_0^2/N_0^2 \ll 1$ but
\begin{equation}\label{rot-grav-scaling}
	\frac{g_b + \tau \, k^2}{f_0^2 \, H} \sim O(1).
\end{equation}
Accordingly, we may neglect the Coriolis parameter in the interior equation \eqref{grav-cap-1} but not at the upper boundary condition \eqref{grav-cap-3}. Substituting the wave solution \eqref{w-wave} into equations \eqref{grav-cap-1}\textendash\eqref{grav-cap-3} yields
\begin{align}
	\label{rot-bound-SL1}
	-\hat w'' + k^2 \, \hat w = \sigma^{-2} \, N^2 \, \hat w \quad &\text{for } z \in (-H,0) \\
	\label{rot-bound-SL2}
	\hat w = 0 \quad &\text{for } z = -H \\
	\label{rot-bound-SL3}
	(g_b + \tau \, k^2)^{-1} \hat w ' =  \sigma^{-2} \left[ \hat w  + \frac{f_0^2}{k^2} (g_b + \tau \, k^2)^{-1} \hat w' \right]\quad &\text{for } z = 0,
\end{align}
where $\sigma=\omega/k$ is the phase speed. Equations \eqref{rot-bound-SL1}\textendash\eqref{rot-bound-SL3} form an eigenvalue problem for the eigenvalue $\lambda = \sigma^{-2}$.

\subsubsection*{Definiteness \& the underlying function space}

As in the previous case, the eigenvalue problem is both right-definite and left-definite, the underlying function space $L^2_\mu$ is given by equation \eqref{gravity-functionspace}, and the appropriate inner product is equation \eqref{grav-inner}.  By right-definiteness, the space $L^2_\mu$, equipped with the inner product \eqref{grav-inner}, is a Hilbert space; thus, all eigenfunctions $\hat W_n$ satisfy
 \begin{equation}
 	\inner{\hat W_m}{\hat W_n} > 0.
 \end{equation}
 By theorem \ref{real-complete}, all eigenvalues $\lambda_n = \sigma_n^{-2}$ are real and the corresponding eigenfunctions $\{\hat W_n\}_{n=0}^\infty$ form an orthonormal basis of the Hilbert space $L^2_\mu$. By proposition \ref{eigenvalue-sign}, all eigenvalues $\lambda_n = \sigma_n^{-2}$ are positive and satisfy equation \eqref{grav-phase-speeds}.

\subsubsection*{Boundary jump discontinuity of the eigenfunctions}

The main difference between the previous non-rotating problem \eqref{non-rot-SL1}\textendash\eqref{non-rot-SL3} and the above problem \eqref{rot-bound-SL1}\textendash\eqref{rot-bound-SL3} is that, in the present problem, if $f_0 \neq 0$ then $d_2 \neq 0$ [see equation \eqref{EigenB2}]. Thus, by equation \eqref{discontinuity}, the eigenfunctions $\hat W_n$ generally have a jump discontinuity at the upper boundary $z=0$ (see figure \ref{F-rotating-eigenfunctions}) and so are not equal to the solutions $\hat w_n$. The eigenfunctions $\hat W_n$ are defined by $\hat W_n(z) = \hat w_n(z)$ for $z \in[-H,0)$ and
\begin{equation}
	\hat W_n(0) = \hat w_n(0) +  \frac{f_0^2}{k^2} (g_b + \tau \, k^2)^{-1} \, \hat w_n'(0)
\end{equation}
[see equation \eqref{discontinuity}]. It is not difficult to show that
\begin{equation}
	\hat W_n(0) \approx 0 \quad \text{ for } n \text{ sufficiently large,}
\end{equation}
as can be seen in figure \ref{F-rotating-eigenfunctions}. 

Physical motion is given by the solutions $\hat w_n$ which are continuous over the closed interval $[-H,0]$. The jump discontinuity in the eigenfunctions $\hat W_n$ does not correspond to any physical motion; instead, the eigenfunctions $\hat W_n$ are convenient mathematical aids used to obtain eigenfunction expansions in the function space $L^2_\mu$.

\subsubsection*{Number of internal zeros of the eigenfunctions} 

Another consequence of $d_2 \neq 0$ is that by, lemma \ref{extra-oscillation}, there are two distinct solutions $\hat w_M$ and $\hat w_{M+1}$ with the same number  of internal zeros (i.e., $M$) in the interval $(-H,0)$. Noting that  
\begin{equation}
	-\frac{b_2}{d_2} = \frac{k^2}{f_0^2}
\end{equation}
the integer $M$ is determined by 
\begin{equation}
	\sigma_0^2 > \sigma_1^2>\dots > \sigma_M^2 > \frac{f_0^2}{k^2} \geq \sigma^2_{M+1} > \dots > 0.
\end{equation}
A smaller $f_0$ or a larger $k$ implies a larger $M$ and hence that $\hat w_M$ and $\hat w_{M+1}$ have a larger number of internal zeros, as shown in figure \ref{F-rotating-eigenfunctions}.

\subsubsection*{Expansion properties}

As in the previous problem, the eigenfunctions are complete in $L^2_\mu$ but overcomplete in $L^2$ due to the additional surface gravity-capillary wave.

Given a twice continuously differentiable function $\chi(z)$ satisfying $\chi(-H)=0$, we define the discontinuous function $X(z)$ by
\begin{equation}
	X(z) = 
	\begin{cases}
		\chi(z) &\quad \textrm{for } z\in [-H,0) \\
		\chi(0) + \frac{f_0^2}{k^2}  \left( g_b + \tau \, k^2 \right)^{-1} \chi'(0)  &\quad \textrm{for } z = 0
	\end{cases}
\end{equation}
as in theorem \ref{uniform}. Then, by theorem \ref{uniform}, we have the expansions
\begin{equation}
	\chi(z) = \sum_{n=0}^\infty \inner{X}{\hat W_n} \, \hat w_n(z) \quad \textrm{and} \quad \chi'(z) = \sum_{n=0}^\infty \inner{X}{\hat W_n}\, w_n'(z).
\end{equation}
Moreover, if $\chi(z)$ is the vertical structure at $t=0$ (and at some wavevector $\vec k$) and we assume $\partial_t w(\vec x,z,t=0)=0$, then the subsequent time-evolution is given by
\begin{equation}
  w(\vec x, z,t) = \sum_{n=0}^\infty \inner{X}{\hat W_n} \, \hat w_n(z) \, \cos\left(\sigma_n k t\right) \,\mathrm{e}^{\mathrm{i}\vec k \cdot \vec x}.
\end{equation}

\section{Quasigeostrophic waves} \label{S-QG}

\subsection{Linear equations}

Linearizing the quasigeostrophic equations about a quiescent background state with an infinitesimally sloping lower boundary, at $z=-H$, and a rigid flat upper boundary, at $z=0$, renders
\begin{align}
	\label{linear-q}
	\partial_t \left[ \lap_z \psi +  \partial_z \left( S^{-1} \, \partial_z \psi \right) \right]  +  \unit z \cdot \left( \grad_z \psi \times \grad_z f \right) &=0 \quad \text{for } z\in(-H,0)\\
	\label{linear-r1}
	\partial_t \left( S^{-1} \, \partial_z \psi \right)  +  \unit z \cdot \left( \grad_z \psi \times f_0 \, \grad_z h \right) &=0 \quad \text{for } z=-H\\
	\label{linear-r2}
	\partial_t \left( S^{-1} \, \partial_z \psi \right) &= 0 \quad \text{for } z=0.
\end{align}
See \cite{rhines_edge_1970}, \cite{charney_oceanic_1981}, \cite{straub_dispersive_1994} for details.
The streamfunction $\psi$ is defined through $\vec u = \unit z \times \grad_z \psi$ where $\vec u$ is the horizontal velocity and $\grad_z = \unit x \, \partial_x + \unit y \, \partial_y$ is the horizontal Laplacian. The stratification parameter $S$ is given by
\begin{equation}
	S(z) = \frac{N^2(z)}{f_0^2},
\end{equation}
where $N^2$ is the buoyancy frequency and $f_0$ is the reference Coriolis parameter. The latitude dependent Coriolis parameter $f$ is defined by
\begin{equation}
	f(y) = f_0 + \beta \, y.
\end{equation}
Finally, $h(\vec x)$ is the height of the topography at the lower boundary and is a linear function of the horizontal position vector $\vec x$. Consistent with quasigeostrophic theory, we assume that topography $h$ is small and so we evaluate the lower boundary condition at $z=-H$ in equation \eqref{linear-r1}.

\subsection{The streamfunction eigenvalue problem} \label{S-QG-stream}

We assume wave solutions of the form
\begin{equation}\label{QG-psi-wave}
	\psi(\vec x, z,t) = \hat \psi (z) \, \mathrm{e}^{\mathrm{i}(\vec k \cdot \vec x - \omega t)}
\end{equation}
where $\vec k = \unit x \, k_x + \unit y \, k_y$ is the horizontal wavevector and $\omega$ is the angular frequency. 
 
We denote by $\Delta \theta_f$ the angle between the horizontal wavevector $\vec k$ and the gradient of Coriolis parameter $\grad_z f$,
\begin{equation} \label{theta_f}
	\sin{(\Delta \theta_f)} = \frac{1}{k \, \beta } \, \unit z \cdot \left(\vec k \times  {\grad_z f }\right),
\end{equation}
where $k= \abs{\vec k}$ is the horizontal wavenumber. Positive angles are measured counter-clockwise relative to $\vec k$. Thus, $\Delta \theta_f>0$ indicates that $\vec k$ points to the right of $\grad_z f$ while $\Delta \theta_f<0$ indicates that $\vec k$ points to the left of $\grad_z f$.

We define the topographic parameter $\alpha$ by
\begin{equation}
	\alpha = \abs{f_0 \, \grad_z h}.
\end{equation}
In analogy with $\Delta \theta_f$, we define the angle $\Delta \theta_h$ by
\begin{equation} \label{theta_i}
	\sin{(\Delta \theta_h)} =  \frac{1}{k \, \alpha} \, \unit z \cdot \left( \vec k \times  {f_0 \grad_z h} \right) 
\end{equation}
with a similar interpretation assigned to $\Delta \theta_h>0$ and $\Delta \theta_h<0$.

Substituting the wave solution \eqref{QG-psi-wave} into the linear quasigeostrophic equations \eqref{linear-q}\textendash \eqref{linear-r2} and assuming that $\alpha \, \sin(\Delta\theta_h) \neq 0$, $\omega \neq 0$, and $k\neq 0$, we obtain
\begin{align}
		\label{QG-psi-l-eigen1}
	- (S^{-1} \, \hat \psi')' + k^2 \, \hat \psi  = \lambda  \, \hat \psi \quad &\text{for } z\in(-H,0)\\
	\label{QG-psi-l-eigen2}
	-\frac{\beta}{\alpha} \, \frac{\sin{(\Delta \theta_f)}}{ \sin{(\Delta \theta_h)}} \, S^{-1} \, \hat \psi' =  \lambda \,  \psi  \quad &\text{for } z=-H \\
	\label{QG-psi-l-eigen3}
	S^{-1}\, \psi' = 0 \quad &\text{for } z=0,
\end{align}
where we have defined the eigenvalue $\lambda$ by
\begin{equation}\label{qg-dispersion}
	\lambda = -\frac{k \, \beta \, \sin{(\Delta \theta_f)}}{\omega}.
\end{equation}
Since $k\neq 0$ then $\lambda=0$ is not an eigenvalue. The above problem \eqref{QG-psi-l-eigen1}\textendash\eqref{QG-psi-l-eigen3} was recently considered in \cite{lacasce_prevalence_2017}.

\begin{figure}
  \centerline{\includegraphics[width=1.\columnwidth]{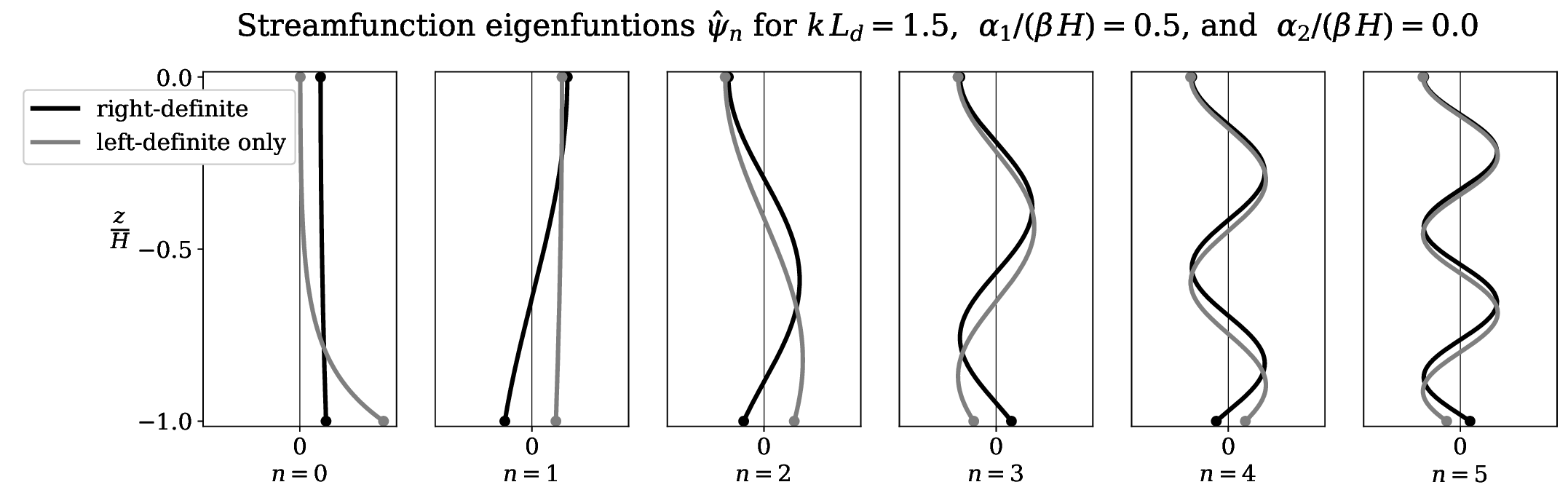}}
  \caption{The streamfunction eigenfunctions $\hat \psi_n$ of the quasigeostrophic eigenvalue problem with a sloping bottom from \S \ref{S-QG-stream}. Two cases are shown. The first is with $\Delta \theta_f = -90^\circ$ and $\Delta \theta_1 = -30^\circ$ and is both right-definite and left-definite. The second is with $\Delta \theta_f = -45^\circ$ and $\Delta \theta_1 = 15^\circ$ and is only left-definite. In the right-definite case the $n$th eigenfunction has $n$ internal zero whereas in the left-definite only case there are two eigenfunctions ($n=0,1$) with no internal zeros.}
  \label{F-QG-psi}
\end{figure}

\subsubsection*{Definiteness \& the underlying function space}

The eigenvalue problem has one $\lambda$-dependent boundary condition and so the underlying function space is 
\begin{equation}
	L^2_\mu \cong L^2 \oplus \C.
\end{equation}
The appropriate inner product is obtained from equations \eqref{dmu_inner} and \eqref{Di}
\begin{align}\label{inner-qg}
	\inner{\varphi}{\phi} =  \frac{1}{H} \left[ \int_{-H}^0 \varphi \, \phi \, \mathrm{d}z + \frac{\alpha}{\beta} \, \frac{\sin\left(\Delta \theta_h\right)}{\sin\left(\Delta \theta_f\right)} \, \varphi(-H) \, \phi(-H) \right]
\end{align} 
where we have introduced the factor $1/H$ for dimensional consistency in eigenfunction expansions. By proposition \ref{right-definite-prop}, the problem is right-definite for horizontal wavevectors $\vec k$ satisfying
\begin{equation}\label{qg-right-definite-condition}
	\frac{\sin{(\Delta \theta_h)}}{\sin{(\Delta \theta_f)}} >0
\end{equation}
and, in such cases, $L^2_\mu$ equipped with the inner product \eqref{dmu_inner} is a Hilbert space. However, $L^2_\mu$ is not a Hilbert space for all wavevectors $\vec k$. By proposition \ref{left-definite-prop}, the problem is left-definite for all wavevectors $\vec k$ and so $L^2_\mu$, equipped with the inner product \eqref{dmu_inner}, is generally a Pontryagin space.

We write $\hat \Psi_n$ for the eigenfunctions and $\hat \psi_n$ for the solutions of equations \eqref{QG-psi-l-eigen1}\textendash \eqref{QG-psi-l-eigen3}. The eigenfunctions $\hat \Psi_n$ are related to the solutions $\hat \psi_n$ by \eqref{L2_mu-element} with boundary values $\hat \Psi_n(0)$ given by equation \eqref{discontinuity}. However, since $c_1 = 1$ and $d_1=0$ in equation \eqref{QG-psi-l-eigen2} [compare with equations \eqref{EigenDiff}\textendash\eqref{EigenB2}] then $\hat \Psi_n = \hat \psi_n$ on the closed interval $[-H,0]$. Thus, the solutions $\psi_n$ are also the eigenfunctions.

With theorem \ref{real-complete}, we deduce that all eigenvalues $\lambda_n$ are real and the corresponding eigenfunctions $\{\hat \psi_n\}_{n=0}^\infty$ form an orthonormal basis for $L^2_\mu$. Orthonormality is defined with respect to the inner product given by equation \eqref{inner-qg} and takes the form
\begin{equation}
	\pm \delta_{mn} = \inner{\hat \psi_m}{\hat \psi_n}
\end{equation}
where we have taken the eigenfunctions $\hat \psi_m$ and $\hat \psi_n$ to be non-dimensional.

\subsubsection*{Properties of the eigenfunctions}

By lemma \ref{extra-oscillation}, the number of internal zeros of the eigenfunctions $\{\hat \psi_n\}_{n=0}^\infty$ depends on the propagation direction and hence [by equation \eqref{qg-right-definite-condition}] on the definiteness of the problem (see figure \ref{F-QG-psi}):
\begin{itemize}
	\item [\emph{1.}] if the problem is right-definite then the $n$th eigenfunction has $n$ internal zeros,
	\item [\emph{2.}] if the problem is not right-definite then both $\psi_0$ and $\psi_1$ have no internal zeros; the remaining eigenfunctions $\psi_n$, for $n>1$, have $n-1$ internal zeros.
\end{itemize}

As the problem is left-definite for all wavevectors $\vec k$, we can use proposition \ref{eigenvalue-sign} to determine the sign of the eigenvalues. 
Proposition \ref{eigenvalue-sign} informs us that
\begin{equation}
	\lambda_n \inner{\hat \psi_n}{\hat \psi_n} > 0.
\end{equation}
In the first case, when the problem is right-definite, all eigenvalues are positive and all eigenfunctions $\hat \psi_n$ satisfy $\inner{\hat \psi_n}{\hat \psi_n} >0$.
In the second case, when the problem is only left-definite, then there is one negative eigenvalue $\lambda_0$ and the corresponding eigenfunction $\hat \psi_0$ satisfies $\inner{\hat \psi_0}{\hat \psi_0} < 0$.
The remaining eigenvalues are positive and their corresponding eigenfunctions satisfy 
$ \inner{\hat \psi_n}{\hat \psi_n} >0$.
In fact, from equation \eqref{qg-dispersion}, we see that waves with $\inner{\hat \psi_n}{\hat \psi_n} >0 $ have westward phase speeds $\omega_n /k <0$ while waves with $\inner{\hat \psi_n}{\hat \psi_n} <0 $ have eastward phase speeds $\omega_n/k >0$.

\subsubsection*{Expansion properties}

The eigenfunctions $\{\hat \psi_n\}_{n=0}^\infty$ are complete in $L^2_\mu$ but overcomplete in $L^2$. Physically, there is now an additional eigenfunction corresponding to a topographic Rossby wave ($n=0$ in figure \ref{F-QG-psi}).

Given a twice continuously differentiable function $\phi(z)$ satisfying $\phi'(0)=0$, then from theorem \ref{uniform}, we have
\begin{equation}
	\phi(z) = \sum_{n=0}^\infty \frac{\inner{\phi}{\hat \psi_n}}{\inner{\hat \psi_n}{\hat \psi_n}} \, \hat \psi_n(z) \quad \textrm{and} \quad \phi'(z) = \sum_{n=0}^\infty \frac{\inner{\phi}{\hat \psi_n}}{\inner{\hat \psi_n}{\hat \psi_n}} \, \hat \psi'_n(z),
\end{equation}
with both series converging uniformly on $[-H,0]$ (note that $\phi$ is not required to satisfy any particular boundary condition at $z=-H$). If the vertical structure at time $t=0$  (and at some wavevector $\vec k$) is given by $\phi$, then the subsequent time-evolution is given by
\begin{equation}\label{qg-time-evolution}
	 \psi(\vec x, z,t) = \sum_{n=0}^\infty \frac{\inner{\phi}{\hat \psi_n}}{\inner{\hat \psi_n}{\hat \psi_n}} \, \hat \psi_n(z) \, \cos\left(\omega_n t\right) \, \mathrm{e}^{\mathrm{i}\vec k \cdot \vec x},
\end{equation}
where the angular frequency $\omega_n$ is given by equation \eqref{qg-dispersion}.

\section{A localized perturbation at the boundary}\label{S-step-forcing}

We now consider a localized perturbation at a dynamically-active boundary; we idealize such a perturbation by a boundary step-function $\Theta_i$ (for $i \in S$) given by
	\begin{equation}\label{Theta}
		\Theta_i(z) =
		\begin{cases}
			1 \quad \text{if } z=z_i\\
			0 \quad \text{otherwise.}
		\end{cases}
	\end{equation}
	Using equation \eqref{expansion}, the series expansion of $\Theta_i$ is found to be
	\begin{equation}\label{theta-expansion}
		\Theta_i 
		=  \frac{1}{D_i} \sum_{n=0}^\infty \frac{\, \Phi_n(z_i)}{\inner{\Phi_n}{\Phi_n}} \, \Phi_n(z).
	\end{equation}
	
For the non-rotating Boussinesq problem of \S \ref{S-Boussinesq-non-rot}, a step-function perturbation with amplitude $w_0$ (at some wavevector $\vec k$) yields the time-evolution
\begin{equation}
	w(\vec x, z,t) = w_0 \left(\frac{g_b + \tau \, k^2}{N_0^2\, H}\right) \sum_{n=0}^\infty \hat w_n(0) \, \hat w_n(z) \, \cos\left(\sigma_n k t \right) \,\mathrm{e}^{\mathrm{i}\vec k \cdot \vec x}.
\end{equation}
Analogously, for the quasigeostrophic problem of \S \ref{S-QG-stream}, a step-function perturbation with amplitude $\psi_0$ (at some wavevector $\vec k$) yields the time-evolution
\begin{equation}
	\psi(\vec x, z,t) = \psi_0 \left[\frac{\alpha \, \sin(\Delta \theta_h) }{H\, \beta \, \sin(\Delta \theta_f)}\right] \sum_{n=0}^\infty \frac{\hat \psi_n(-H)}{\inner{\hat\psi_n}{\hat \psi_n}} \, \hat \psi_n(z) \, \cos\left(\omega_n t\right) \,\mathrm{e}^{\mathrm{i} \vec k \cdot \vec x}.
\end{equation}
That both the above series converge to a step-function at $t=0$ (and $\vec x = \vec 0$) is confirmed by theorem \ref{expansion-theorem} along with theorem 2 in \citet{fulton_two-point_1977}.

We thus see that a step-function perturbation induces wave motion with an amplitude that is proportional to the boundary-confined restoring force (at wavevector $\vec k$). Moreover, the amplitude of each constituent wave in the resulting motion is proportional to the projection of that wave onto the dynamically-active boundary.

\section{Summary and conclusions}\label{S-conclusion}

We have developed a mathematical framework for the analysis of three-dimensional wave problems with dynamically-active boundaries (i.e., boundaries where time derivatives appear in the boundary conditions). The resulting waves have vertical structures that depend on the wavevector $\vec k$: For Boussinesq gravity waves, the dependence is only through the wavenumber $k$ whereas the dependence for quasigeostrophic Rossby waves is on both the wavenumber $k$ and the propagation direction $\vec k/k$. Moreover, the vertical structures of the waves are complete in a space larger than $L^2$, namely, they are complete in $L^2_\mu \cong L^2 \oplus \C^s$ where $s$ is the number of dynamically active boundaries (and the number of boundary-trapped waves). Each dynamically active boundary contributes an additional boundary-trapped wave and hence an additional degree of freedom to the problem. Mathematically, the presence of boundary-trapped waves allows us to expand a larger collection of functions (with a uniformly convergent series) in terms of the modes. The resulting series are term-by-term differentiable and the differentiated series converges uniformly. In fact, the normal modes have the intriguing property converging pointwise to functions with finite jump discontinuities at the boundaries, a property related to their ability to expand distributions in the \cite{bretherton_critical_1966} ``$\delta$-function formulation'' of a physical problem. By considering a step-function perturbation at a dynamically-active boundary, we find that the subsequent time-evolution consists of waves whose amplitude is proportional to their projection at the dynamically-active boundary. Within the mathematical formulation is a qualitative oscillation theory relating the number of internal zeros of the eigenfunctions to physical quantities; indeed, for the quasigeostrophic problem, the number of zeros of the topographic Rossby wave depends on the propagation direction while, for the rotating Boussinesq problem, the ratio of the Coriolis parameter to the horizontal wavenumber determines at which integer $M$ we obtain two modes with $M$ zeros.

%%% NEW ADD
Our results also clarify the difference between the traditional quasigeostrophic baroclinic modes and the the $L^2\oplus \C^2$ eigenfunctions of \cite{smith_surface-aware_2012}. Namely, the series expansion of a function in terms of the \cite{smith_surface-aware_2012} eigenfunctions has a term-by-term derivative that converges uniformly over the whole interval regardless of the boundary conditions satisfied by the function. In contrast, an eigenfunction expansion in terms of the baroclinic modes only has this property if the function satisfies the same boundary conditions as the baroclinic modes. One consequence is the following. Suppose we expand an arbitrary quasigeostrophic state, with boundary buoyancy anomalies, in terms of the baroclinic modes. The presence of these boundary buoyancy anomalies implies that this state does not satisfy the same boundary conditions as the baroclinic modes. The resulting series expansion in term of the baroclinic modes is then not differentiable at the boundaries. We are thus unable to recover the value of the boundary buoyancy anomalies from the series expansion and so we have lost information in the expansion process.  This loss of information does not occur with  $L^2\oplus \C^2$ expansions. 

Normal mode decompositions of quasigeostrophic motion play an important role in physical oceanography \cite[e.g.,][]{wunsch_vertical_1997,lapeyre_what_2009,lacasce_prevalence_2017}. Other applications include the extension of equilibrium statistical mechanical calculations \cite[e.g.,][]{bouchet_statistical_2012,venaille_catalytic_2012} to three-dimensional systems with dynamically-active boundaries. Moreover, the mathematical framework developed here is useful for the development of weakly non-linear wave turbulence theories \cite[e.g.,][]{fu_nonlinear_1980,smith_scales_2001,scott_wave_2014} in systems with both internal and boundary-trapped waves.

\begin{acknowledgments}
	I am grateful to Stephen Griffies for numerous readings, helpful comments, and advice offered throughout this project. This report was prepared by Houssam Yassin under award NA18OAR4320123 from the National Oceanic and Atmospheric Administration, U.S. Department of Commerce. The statements, findings, conclusions, and recommendations are those of the author and do not necessarily reflect the views of the National Oceanic and Atmospheric Administration, or the U.S. Department of Commerce. The data that supports the findings of this study are available within the article and its supplementary material.
\end{acknowledgments}

\appendix

\section{Additional properties of the eigenvalue problem}\label{S-additional-properties}

\subsection{Construction of $L^2_\mu$}\label{L2-mu-construction}

First, define the weighted Lebesgue measure $\sigma$ by
\begin{equation}
	\sigma([a,b]) = \int_a^b r \, \mathrm{d}z \quad \text{where } a,b\in[z_1,z_2].
\end{equation}
The measure $\sigma$ induces the differential element
\begin{equation}
	\mathrm{d}\sigma(z) = r(z)\,\mathrm{d}z
\end{equation}
and is the measure associated with $L^2$ [see equations \eqref{L2-equal} and \eqref{L2-inner}].

Now, for $i\in S$ [see equation \eqref{S-set}], define the pure point measure $\nu_i$ by \cite[e.g.,][section I.4, example 2]{reed_methods_1980}
\begin{equation}
	\nu_i([a,b]) = 
		\begin{cases}
			D_i^{-1} \ \ \ &\text{if } z_i \in [a,b] \\
			0 \ \ \ &\text{otherwise,}
		\end{cases}
\end{equation}
where $D_{i}$ is the combination of boundary condition coefficients given by equation \eqref{Di}. 
The pure point measure $\nu_i$ induces the differential element 
\begin{equation}
	\mathrm{d}\nu_i(z) = D_i^{-1}\, \delta(z-z_i) \, \mathrm{d}z,
\end{equation}
where $\delta(z)$ is the Dirac distribution.

Consider now the space $L^2_{\nu_i}$ of ``functions'' $\phi$ satisfying  
\begin{equation}
	\abs{ \intz \abs{\phi}^2\, \mathrm{d}\nu_i} = \abs{D_i^{-1}} \, \intz \abs{\phi}^2 \, \delta(z-z_i) \, \mathrm{d}z = \abs{D_i^{-1}} \,\abs{\phi(z_i)}^2 < \infty.
\end{equation}
Elements of $L^2_{\nu_i}$ are not functions, but rather equivalence classes of functions. Two functions, $\phi$ and $\psi$, on the interval $[z_1,z_2]$ are equivalent in $L^2_{\nu_i}$ if $\phi(z_i) = \psi(z_i)$. In particular, $L^2_{\nu_i}$ is a one-dimensional vector space and is hence isomorphic to the field of complex numbers $\C$
\begin{equation}
	L^2_{\nu_i} \cong \C.
\end{equation}

Now define the measure $\mu$ by
\begin{equation}\label{mu}
	\mu = \sigma + \sum_{i\in S} \nu_i
\end{equation}
with an induced differential element of
\begin{equation} \label{dmu-appendix}
	\mathrm{d}\mu(z) = \left[ r(z) + \sum_{i\in S} D_i^{-1} \, \delta(z-z_i) \right] \mathrm{d}z.
\end{equation}
Then $L^2_\mu$ is the space of equivalence classes of functions that are square-integrable with respect to the measure $\mu$.

Since the measures $\sigma$ and $\nu_i$, for $i\in S$, are mutually singular, we have \cite[][section II.1, example 5]{reed_methods_1980}
\begin{equation}\label{isomorphism}
	L^2_\mu \cong L^2 \oplus \sum_{i\in S} L^2_{\nu_i} \cong L^2 \oplus \C^s
\end{equation}
from which we see that $L^2_\mu$ is ``larger'' by $s$ dimensions.

\subsection{The eigenvalue problem in $L^2_\mu$}\label{S-eigen-in-L2-mu}

We construct here an operator formulation of \eqref{EigenDiff}\textendash\eqref{EigenB2} as an eigenvalue problem in the Pontryagin space $L^2_\mu$.

Define the differential operator $\ell$ acting on a function $\phi$ by
\begin{equation}
	\ell \, \phi = \frac{1}{r} \, \left[  (p \, \phi')' - q \, \phi \right].
\end{equation}
We also define the following boundary operators for $i \in S$, 
\begin{align}
	\Bc_i \phi &= \left[a_i \, \phi(z_i) - b_i \, (p \, \phi')(z_i)\right]\\
	\Cc_i  \phi &= \left[c_i \, \phi(z_i) - d_i \, (p \, \phi')(z_i)\right].
\end{align}
Let $\Phi$ be an element of $L^2_\mu$, as in equation \eqref{L2_mu-element}, with boundary values $\Phi(z_i) = \Cc_i \phi$ for $i\in S$ and equal to $\phi$ elsewhere.
We then define the operator $\Lc$, acting on functions $\Phi$, by
\begin{equation}
	\Lc \, \Phi = 
	\begin{cases}
		-\ell \, \phi  \quad &\text{for } z\in(z_1,z_2)\\
		-\Bc_i \,  \phi  \quad &\text{for } z=z_i \text{ where }i\in S
	\end{cases}
\end{equation}
with a domain $D(\Lc)\subset L^2_\mu$ defined by
\begin{equation}
\begin{aligned}
	D(\Lc) = \{ \Phi \in L^2_\mu \ | & \ \phi \text{ is continuously differentiable, } \ell \, \phi \in L^2,   \ \Phi(z_i) = \Cc_i \, \phi \\\text{ for } i\in S  & \text{ and } \Bc_i \phi = 0 \text{ for } i \in \{1,2\}\setminus S  \}.
\end{aligned}
\end{equation}
Recall that $S$ contains indices of the $\lambda$-dependent boundary conditions, and therefore, $\{1,2\}\setminus S$ contains the indices of the $\lambda$-independent boundary conditions. 

Then, on the subspace $D(\Lc)$ of $L^2_\mu$, the eigenvalue problem \eqref{EigenDiff}\textendash\eqref{EigenB2} may be written as 
\begin{equation}\label{Op-Eigenproblem-app}
	\Lc \, \Phi = \lambda \, \Phi.
\end{equation} 
As shown in \cite{russakovskii_operator_1975,russakovskii_matrix_1997}, $\Lc$ is a self-adjoint operator in the space $L^2_\mu$.

There is a natural quadratic form $Q$, induced by the eigenvalue problem \eqref{EigenDiff}\textendash\eqref{EigenB2}, given by
\begin{equation}\label{quadratic_form}
	Q(\Phi,\Psi) = \inner{\Phi}{\Lc \, \Psi}.
\end{equation}
For elements $\Phi,\Psi \in D(\Lc)$, we obtain
\begin{equation}
\begin{aligned}
		Q(\Phi,\Psi)  &= \intz \left[p \, \ov{\phi'} \, \psi' + q \, \ov{\phi} \, \psi \right] \, \mathrm{d}z 
	+ \sum_{i\in\{1,2\}\setminus S} (-1)^{i+1} \, \frac{a_i}{b_i} \,\ov{\phi(z_i)} \, \psi(z_i) \, \\
	&\quad - \sum_{i\in S} \frac{1}{D_i}
	\ov{
	\left(
	\begin{matrix}
		\psi(z_i) \\
		-(p \, \psi')(z_i)
	\end{matrix}
	\right)
	} \cdot
	\left(
	\begin{matrix}
		a_i \, c_i & a_i \, d_i \\
		a_i \, d_i & b_i \, d_i
	\end{matrix}
	\right)
	\left(
	\begin{matrix}
		\phi(z_i) \\
		-(p \, \phi')(z_i)
	\end{matrix}
	\right)
\end{aligned}
\end{equation}
for $b_i \neq 0$ for $i \in \{1,2\} \setminus S$. If $b_i=0$ for $i \in \{1,2\} \setminus S$ then we replace the term $a_i/b_i$ with zero. 

To develop the reality and completeness theorem \ref{real-complete}, we provide the following definitions.

\begin{definition}[Right-definite]
	The eigenvalue problem \eqref{EigenDiff}\textendash\eqref{EigenB2} is said to be right-definite if $L^2_\mu$ is a Hilbert space or, equivalently, if
	\begin{equation}
		\inner{\Phi}{\Phi} > 0 
	\end{equation} 
	for all non-zero $\Phi \in L^2_\mu$.
\end{definition}

\begin{definition}[Left-definite]
	The eigenvalue problem \eqref{EigenDiff}\textendash\eqref{EigenB2} is said to be left-definite if 
	\begin{equation}\label{left-definite-inequality}
		Q(\Phi,\Phi) \geq 0
	\end{equation}
	for all $\Phi \in D(\Lc)$.
\end{definition}

One can then prove propositions \ref{right-definite-prop} and \ref{left-definite-prop} through straightforward manipulations.

\subsection{Properties of eigenfunction expansions}\label{S-more-eigen-expansions}

The following theorem features some of the novel properties of the basis $\{\Phi_n\}_{n=0}^\infty$ of $L^2_\mu$. Theorem \ref{expansion-theorem} below is a generalization of a theorem first formulated, in the right-definite case, by \citet{walter_regular_1973} and \cite{fulton_two-point_1977}.

\begin{theorem}[Eigenfunction expansions]\label{expansion-theorem} Let $\{\Phi_n\}_{n=0}^\infty$ be the set of eigenfunctions of the eigenvalue problem \eqref{EigenDiff}\textendash\eqref{EigenB2}. Then the following properties hold.
\begin{itemize}
	\item[(i)] Null series:  For $i \in S$, we have 
	\begin{equation}\label{null}
		0 = D_i^{-1} \, \sum_{n=0}^\infty \frac{1}{\inner{\Phi_n}{\Phi_n}} \, \Phi_n(z_i) \, \phi_n(z) 
	\end{equation}
	with equality in the sense of $L^2$.
	\item[(ii)] Unit series: For $i \in S$, we have
	\begin{equation}\label{unit}
		1 = D_i^{-1} \, \sum_{n=0}^\infty \frac{1}{\inner{\Phi_n}{\Phi_n}} \, \abs{\Phi_n(z_i)}^2.
	\end{equation}
	\item[(iii)] $L^2$-expansion: Let $\psi \in L^2$, then
	\begin{equation}\label{L2-expansion}
		\psi =  \sum_{n=0}^\infty \frac{1}{\inner{\Phi_n}{\Phi_n}} \,\left(\intz \ov{\psi} \, \phi_n \, r \, \mathrm{d}z \right) \, \phi_n.
	\end{equation}
	with equality in the sense of $L^2$.
	\item[(iv)] Interior-boundary orthogonality: Let $\psi \in L^2$, then for $i \in S$, we have
	\begin{equation}\label{int-bound-orth}
		0  = \sum_{n=0}^\infty \frac{1}{\inner{\Phi_n}{\Phi_n}}  \left(\intz \ov{\psi} \, \phi_n \, r \, \mathrm{d}z \right) \Phi_n(z_i).
	\end{equation}
\end{itemize}
\end{theorem}
\begin{proof}
	The proof is similar to the proof of corollary 1.1 in \cite{fulton_two-point_1977}.
\end{proof}

\subsection{Pointwise convergence and Sturm-Liouville series}\label{S-pointwise}

Theorem 3 in \cite{fulton_two-point_1977} states that the $\Phi_n$ series expansion \eqref{expansion} behaves like a Fourier series in the interior of the interval $(z_1,z_2)$ (see appendix \ref{S-math-app} for why this theorem applies in the left-definite case). Since the expansions \eqref{expansion} and \eqref{expansion-phi} in terms of $\Phi_n$ and $\phi_n$ are equal in the interior, then the above theorem applies to the $\phi_n$ series \eqref{expansion-phi} as well. It is at the boundaries points, $z=z_1,z_2$, where the novel behaviour of the series expansions \eqref{expansion} and \eqref{expansion-phi} appears.

For traditional Sturm-Liouville expansions [with eigenfunctions of problem \eqref{EigenDiff}-\eqref{EigenB2} with $c_i,d_i = 0$ for $i=1,2$], eigenfunction expansions behave like the analogous Fourier series on $[z_1,z_2]$ [page 16 in \cite{titchmarsh_eigenfunction_1962}  or chapter 1, section 9, in \cite{levitan_introduction_1975}]. In particular, for a twice continuously differentiable function $\psi$, the eigenfunction expansion of $\psi$ converges uniformly to $\psi$ on $[z_1,z_2]$ so long as the eigenfunctions $\phi_n$ do not vanish at the boundaries. If the eigenfunctions vanish at one of the boundaries, then we only obtain uniform convergence if $\psi$ vanishes at the corresponding boundary as well \citep[][section 22]{brown_fourier_1993}. Under these conditions, the resulting expansion will be differentiable in the interior of the interval, $(z_1,z_2)$, but not at the boundaries $z=z_1,z_2$ [see chapter 8, section 3, in \citet{levitan_introduction_1975} for the equiconvergence of differentiated Sturm-Liouville series with Fourier series and see section 23 in \citet{brown_fourier_1993} for the convergence behaviour of differentiated Fourier series].

Returning to the case of eigenfunction expansions for the eigenvalue problem \eqref{EigenDiff}\textendash\eqref{EigenB2} with $\lambda$-dependent boundaries, the following theorem provides pointwise (as well as uniform, in the case $d_i\neq 0$) convergence conditions for the $\phi_n$ series \eqref{expansion-phi}.

\begin{theorem}[Pointwise convergence]\label{pointwise}
	Let $\psi$ be a twice continuously differentiable function on the interval $[z_1,z_2]$ satisfying any $\lambda$-independent boundary conditions in the eigenvalue problem \eqref{EigenDiff}\textendash\eqref{EigenB2}. Define the function $\Psi$ on $[z_1,z_2]$ by
	\begin{equation}
		\Psi(z) = 
		\begin{cases}
			\Psi(z_i) \quad &\textrm{at } z=z_i, \textrm{ for } i \in S,\\
			\psi(z) \quad &\textrm{otherwise}.
		\end{cases}
	\end{equation}
	where $\Psi(z_i)$ are constants for $i\in S$ (the $\lambda$-dependent boundaries). Then we have the following.
	\begin{itemize}
		\item [(i)] If $d_i \neq 0$ for $i \in S$, then the $\phi_n$ series expansion \eqref{expansion-phi} converges uniformly to $\psi(z)$ on the closed interval $[z_1,z_2]$,
		\begin{equation}
			\sum_{n=0}^\infty \frac{\inner{\Psi}{\Phi_n}}{\inner{\Phi_n}{\Phi_n}} \, \phi_n(z) = \psi(z).
		\end{equation} 
		 Furthermore, for the differentiated series, we have
			\begin{equation}
			\sum_{n=0}^\infty \frac{\inner{\Psi}{\Phi_n}}{\inner{\Phi_n}{\Phi_n}} \, \phi_n'(z) = 
			\begin{cases}
				\left( c_i \, \psi(z_i)- \Psi(z_i) \right)/ d_i \quad &\textrm{at } z=z_i, \textrm{ for } i \in S\\
				\psi'(z) \quad &\textrm{otherwise}.
			\end{cases}
			\end{equation} 
		\item [(ii)] If $d_i =0$, then we have
			\begin{equation}
				\sum_{n=0}^\infty \frac{\inner{\Psi}{\Phi_n}}{\inner{\Phi_n}{\Phi_n}} \, \phi_n=
				\begin{cases}
					\Psi(z_i)/c_i \quad &\textrm{at } z=z_i, \textrm{ for } i \in S\\
					\psi(z) \quad &\textrm{otherwise}.
				\end{cases}
			\end{equation}
	\end{itemize}
\end{theorem}
\begin{proof}
	This theorem is a generalization of corollary 2.1 in \cite{fulton_two-point_1977}. We provide the extension of the corollary to the left-definite problem in appendix \ref{S-A-left-fulton}.
\end{proof}

The $\Phi_n$ series \eqref{expansion} converges to $\Psi(z_i)$ at $z=z_i$ 
for $i \in S$ (i.e., at $\lambda$-dependent boundaries) but otherwise behaves as in theorem \ref{pointwise}. 

\section{Literature survey and mathematical proofs}\label{S-math-app}

\subsection{Literature survey}\label{S-lit-review}

There is an extensive literature associated with the eigenvalue problem \eqref{EigenDiff}\textendash\eqref{EigenB2} with $\lambda$-dependent boundary conditions \citep[see][ and citations within]{schafke_s-hermitesche_1966,fulton_two-point_1977}. One can use the $S$-hermitian theory of \cite{schafke_s-hermitesche_1965,schafke_s-hermitesche_1966,schafke_s-hermitesche_1968} to show that one obtains real eigenvalues when the problem is either right-definite or left-definite (see \S \ref{math-section}) but completeness results in $L^2_\mu$ are unavailable in this theory.

The right-definite theory is well-known \citep{evans_non-self-adjoint_1970,walter_regular_1973,fulton_two-point_1977}. In particular, \cite{fulton_two-point_1977} applies the residue calculus techniques of \cite{titchmarsh_eigenfunction_1962} to the right-definite problem and, in the process, extends some well-known properties of Fourier series to eigenfunction expansions associated with \eqref{EigenDiff}\textendash\eqref{EigenB2}. A recent Hilbert space approach to the right-definite problem, in the context of obtaining a projection basis for quasigeostrophic dynamics, is given by \cite{smith_surface-aware_2012}.

The left-definite problem is less examined. As we show in this article, the eigenvalue problem is naturally formulated in a Pontryagin space, and, in such a setting, one can prove, in the left-definite case, that the eigenvalues are real and that the eigenfunctions form a basis for the underlying function space. We prove this result, stated in theorem \ref{real-complete}, in appendix \ref{S-real-proof}.

With these completeness results, we may apply the residue calculus techniques of \cite{titchmarsh_eigenfunction_1962} to extend the results of \cite{fulton_two-point_1977} to the left-definite problem. Indeed, \cite{fulton_two-point_1977} uses a combination of Hilbert space methods as well as residue calculus techniques to prove various convergence results for the right-definite problem. However, only theorem 1 of \cite{fulton_two-point_1977} makes use of Hilbert space methods. If we extend Fulton's theorem 1 to the left-definite problem, then all the results of \cite{fulton_two-point_1977} will apply equally to the left-definite problem. A left-definite analogue of theorem 1 of \cite{fulton_two-point_1977}, along with its proof, is given in appendix \ref{S-A-left-fulton}.

\subsection{A Pontryagin space theorem} \label{S-A-Pontryagin}

A Pontryagin space $\Pi_\kappa$, for a finite non-negative integer $\kappa$, is a Hilbert space with a $\kappa$-dimensional subspace of elements satisfying 
\begin{align}
	\inner{\phi}{\phi}<0.
\end{align}
An introduction to the theory of Pontryagin spaces can be found in \cite{iohvidov_spectral_1960} as well as in the monograph of \cite{bognar_indefinite_1974}. Another resource is the monograph of \cite{azizov_linear_1989} on linear operators in indefinite inner product spaces.

Pontryagin spaces admit a decomposition 
\begin{equation}
	\Pi_{\kappa} = \Pi^+ \oplus \Pi^-
\end{equation}
into orthogonal subspaces $(\Pi^+, +\inner{\cdot}{\cdot})$ and $(\Pi^-, -\inner{\cdot}{\cdot})$. Moreover, one can associate with a Pontryagin space $(\Pi_\kappa,\inner{\cdot}{\cdot})$ a corresponding Hilbert space $(\Pi, \inner{\cdot}{\cdot{}}_+)$ where the positive-definite inner product $\inner{\cdot}{\cdot}_+$ is defined by
\begin{equation}\label{induced_hilbert_inner}
	\inner{\phi}{\psi}_+ = \inner{\phi_+}{\psi_+} - \inner{\phi_-}{\psi_-},  \quad \phi,\psi \in \Pi,
\end{equation}
where $\phi=\phi_++\phi_-$ and $\psi=\psi_++\psi_-$, with $\phi_{\pm},\psi_{\pm}\in \Pi^{\pm}$ \citep{azizov_linear_1981}. 

As a prerequisite to proving theorem \ref{real-complete}, we require the following.

\begin{theorem}[Positive compact Pontryagin space operators]\label{positive-compact}

	Let $\Ac$ be a positive compact operator in a Pontryagin space $\Pi_\kappa$ and suppose that $\lambda=0$ is not an eigenvalue. Then all eigenvalues are real and the corresponding eigenvectors form an orthonormal basis for $\Pi_\kappa$. There are precisely $\kappa$ negative eigenvalues and the remaining eigenvalues are positive. Moreover, positive eigenvalues have positive eigenvectors and negative eigenvalues have negative eigenvectors.
\end{theorem}
\begin{proof}
	By theorem VII.1.3 in \citet{bognar_indefinite_1974} the eigenvalues are all real. Moreover, since $\lambda=0$ is not an eigenvalue, then all eigenspaces are definite \citep[][theorem VII.1.2]{bognar_indefinite_1974} and hence all eigenvalues are semi-simple \citep[][lemma II.3.8]{bognar_indefinite_1974}.

	Since $\Ac$ is a compact operator and $\lambda=0$ is not an eigenvalue, then the span of the generalized eigenspaces is dense in $\Pi_\kappa$ \cite[][lemma 4.2.14]{azizov_linear_1989}. Since all eigenvalues are semi-simple, then all generalized eigenvectors are eigenvectors and so the span of the eigenvectors is dense in $\Pi_\kappa$. Orthogonality of eigenvectors can be shown as in a Hilbert space.
	
	Let $\lambda$ be an eigenvalue and $\phi$ the corresponding eigenvector. By the positivity of $\Ac$, we have
	\begin{equation}
		\inner{\Ac \, \phi}{\phi} = \lambda \inner{\phi}{\phi} \geq 0.
	\end{equation}
	Since all eigenspaces are definite, it follows that positive eigenvectors must correspond to positive eigenvalues and negative eigenvectors must correspond to negative eigenvalues.
	
	Finally, by theorem IX.1.4 in \cite{bognar_indefinite_1974}, any dense subset of $\Pi_\kappa$ must contain a negative-definite $\kappa$ dimensional subspace. Consequently, there are $\kappa$ negative eigenvectors and hence $\kappa$ negative eigenvalues.
\end{proof}

\subsection{Proof of theorem \ref{real-complete}} \label{S-real-proof}

\begin{proof}
The proof for the left-definite case is essentially the standard proof \citep[e.g.,][section 5.10]{debnath_introduction_2005} with theorem \ref{positive-compact} substituting for the Hilbert-Schmidt theorem. We give a general outline nonetheless.

 First, it is well-known that $\Lc$ is self-adjoint in $L^2_\mu$ \citep[e.g.,][]{russakovskii_operator_1975,russakovskii_matrix_1997}. Since $\lambda =0$ is not an eigenvalue, then the inverse operator $\Lc^{-1}$ exists and is an integral operator on $L^2_\mu$. For an explicit construction, see section 4 in \citet{walter_regular_1973}, \cite{fulton_two-point_1977}, and \cite{hinton_expansion_1979}. The eigenvalue problem for $\Lc$, equation \eqref{Op-Eigenproblem}, is then equivalent to
\begin{equation}
	\Lc^{-1} \, \phi = \lambda^{-1} \, \phi
\end{equation}
and both problems have the same eigenfunctions.

The operator $\Lc^{-1}$ is a positive compact operator and so satisfies the requirements of theorem \ref{positive-compact}. Application of theorem \ref{positive-compact} to $\Lc^{-1}$ then assures that all eigenvalues $\lambda_n$ are real, the eigenfunctions form an orthonormal basis for $L^2_\mu$, and the sequence of eigenvalues $\{\lambda_n\}_{n=0}^\infty$ is countable and bounded from below.

The claim that the eigenvalues are simple is verified in \citet{binding_left_1999} for the left-definite problem. Alternatively, an argument similar to that of \cite{fulton_two-point_1977} and \citep[][page 12]{titchmarsh_eigenfunction_1962} can be made to prove the simplicity of the eigenvalues.
\end{proof}

\subsection{Extending \cite{fulton_two-point_1977} to the left-definite problem}\label{S-A-left-fulton}

The following is a left-definite analogue of theorem 1 in \cite{fulton_two-point_1977}. The proof is almost identical to the right-definite case \citep{fulton_two-point_1977, hinton_expansion_1979} with minor modifications. Essentially, since $\inner{\Psi}{\Psi}$ can be negative, we must replace these terms in the inequalities below with the induced Hilbert space inner product $\inner{\Psi}{\Psi}_+$ given by equation \eqref{induced_hilbert_inner}. Our $L^2_\mu$ Green's functions $G$ corresponds to $\tilde G$ in \citet{hinton_expansion_1979}.

\begin{theorem}[A left-definite extension of Fulton's theorem 1]
Let $\Psi \in L^2_\mu$ be defined on the interval $[z_1,z_2]$ by
	\begin{equation}
		\Psi(z) = 
		\begin{cases}
			\Psi(z_i) \quad &\textrm{at } z=z_i, \textrm{ for } i \in S,\\
			\psi(z) \quad &\textrm{otherwise},
		\end{cases}
	\end{equation}
where $\psi \in L^2$ and $\Psi(z_i)$ are constants for $i \in S$. The eigenfunctions $\Phi_n$ are defined similarly (see \S \ref{math-section}). 
\begin{itemize}
	\item[(i)] Parseval formula: For $\Psi \in L^2_\mu$, we have
	\begin{equation}\label{parseval-formula}
		\inner{\Psi}{\Psi} = \sum_{n=0}^\infty \frac{\abs{\inner{\Psi}{\Phi_n}}^2}{\inner{\Phi_n}{\Phi_n}}.
	\end{equation}
	\item[(ii)] For $\Psi \in D(\Lc)$, we have
	\begin{equation} \label{expansion-fulton-theorem1}
		\Psi = \sum_{n=0}^\infty \frac{\inner{\Psi}{\Phi_n}}{\inner{\Phi_n}{\Phi_n}}  \Phi_n.
	\end{equation}
	with equality in the sense of $L^2_\mu$. Moreover, we have 
	\begin{equation}\label{expansion-int-series}
		\psi = \sum_{n=0}^\infty \frac{\inner{\Psi}{\Phi_n}}{\inner{\Phi_n}{\Phi_n}}  \phi_n,
	\end{equation}
	which converges uniformly and absolutely for $z \in[z_1,z_2]$ and may be differentiated term-by-term, with the differentiated series converging uniformly and absolutely to $\psi'$ for $z \in [z_1,z_2]$. The boundaries series
	\begin{equation}\label{boundary_series_appendix}
		\Psi(z_i)  = \sum_{n=0}^\infty \frac{\inner{\Psi}{\Phi_n}}{\inner{\Phi_n}{\Phi_n}}  \Phi_n(z_i),
	\end{equation}
	for $i \in S$, is absolutely convergent.
\end{itemize}
\end{theorem}

\begin{proof}
	The Parseval formula \eqref{parseval-formula} is a consequence of the completeness of the eigenfunctions $\{\Phi_n\}_{n=0}^\infty$ in $L^2_\mu$, given by theorem \ref{real-complete}, and theorem IV.3.4 in \cite{bognar_indefinite_1974}. Similarly, the expansion \eqref{expansion-fulton-theorem1} is also due to completeness of the eigenfunctions. 
	
	We first prove that the series  \eqref{expansion-int-series} converges uniformly and absolutely for $z \in[z_1,z_2]$. We begin with the identity
	\begin{equation}\label{phi-green}
		\phi_n(z) = (\lambda - \lambda_n) \inner{G(z,\cdot,\lambda)}{\Phi_n}
	\end{equation} 
	where $\lambda \in \C$ is not an eigenvalue of $\Lc$, and $G$ is the $L^2_\mu$ Green's function [see equation (8) in \cite{hinton_expansion_1979}]. Then 
	\begin{equation}\label{green_convergence}
		\sum_{n=0}^\infty \lambda_n \frac{\abs{\phi_n}^2}{\abs{\lambda -\lambda_n}^2} = \sum_{n=0}^\infty \lambda_n \abs{\inner{G(z,\cdot,\lambda)}{\Phi_n}}^2 \leq \inner{G(z,\cdot,\lambda)}{\Lc G(z,\cdot,\lambda)}_+ \leq B_1(\lambda)
	\end{equation}
	where $\inner{\cdot}{\cdot}_+$ is the induced Hilbert space inner product given by equation \eqref{induced_hilbert_inner} and $B_1(\lambda)$ is a $z$ independent upper bound \cite[equation 9 in][]{hinton_expansion_1979}. In addition, since $\Psi \in D(\Lc)$, then $\inner{\Lc \Psi}{\Lc \Psi}_+ < \infty$. Thus, we obtain 
	\begin{equation}\label{Lc_psi_2_convergence}
		\sum_n \lambda_n^2 \abs{\inner{\Psi}{\Phi_n}}^2 = \inner{\Lc \Psi}{\Lc \Psi}_+ < \infty.
	\end{equation}
	
	The uniform and absolute convergence of \eqref{expansion-int-series} follows from
	\begin{align}
		\sum_{n=0}^\infty \abs{\frac{\inner{\Psi}{\Phi_n}}{\inner{\Phi_n}{\Phi_n}}  \phi_n} &= \sum_{n=0}^\infty \abs{ \left( \frac{\phi_n}{\lambda - \lambda_n}\right) \left(\lambda - \lambda_n \right) \frac{\inner{\Psi}{\Phi_n} }{\inner{\Phi_n}{\Phi_n}}} \\
			&\leq \sqrt{\left(\sum_{n=0}^\infty \abs{\frac{\phi_n}{\lambda-\lambda_n}}^2\right) \left( \sum_{n=0}^\infty \abs{\lambda-\lambda_n}^2 \abs{\inner{\Psi}{\Phi_n}}^2 \right)}
	\end{align}
	along with equations \eqref{green_convergence} and \eqref{Lc_psi_2_convergence}. The absolute convergence of the boundary series \eqref{boundary_series_appendix} follows as well.
	
	To show that the series \eqref{expansion-int-series} is term-by-term differentiable, it is sufficient to show that the differentiated series converges uniformly for $z \in [z_1,z_2]$ \citep[][section 6.14, theorem 33]{kaplan_advanced_1993}. The proof of the unform convergence of the differentiated series follows from the identity \citep{hinton_expansion_1979}
	\begin{equation}
		\frac{\phi_n'}{\lambda - \lambda_n} = \frac{\mathrm{d}}{\mathrm{d}z} \inner{G(z,\cdot,\lambda)}{\Phi_n} = \inner{\partial_z G(z,\cdot,\lambda)}{\Phi_n}.
	\end{equation}
	and a similar argument.
	
\end{proof}

% If in two-column mode, this environment will change to single-column format so that long equations can be displayed. 
%\begin{widetext}
%$$\mbox{put long equation here}$$
%\end{widetext}

% Figures should be put into the text as floats. 
% Use the graphics or graphicx packages (distributed with LaTeX2e).
% Put the label that you will use with \ref{} command in the braces of the \label{} command.
%
% \begin{figure}
% \end{figure}

% \begin{table}
% \end{tabular}
% \end{table}

%\begin{acknowledgments}
% Put your acknowledgments here.
%\end{acknowledgments}

\bibliography{references}

\end{document}